\documentclass[sigconf,nonacm]{acmart}
\usepackage{popets}
\usepackage{enumitem}
\usepackage{tabularx}

\setcopyright{none}
\theoremstyle{acmdefinition}

\begin{document}

\title[Picture the Epsilon]{Picture the Epsilon: Pursuing Identity-Level Privacy Guarantees for Images}

\author{Arman Zareian Jahromi}
\email{zareian@ksu.edu}
\affiliation{%
  \institution{Kansas State University}
  \city{Manhattan}
  \state{Kansas}
  \country{United States}}

\author{Vishnu Bondalakunta}
\email{vishnub@ksu.edu}
\affiliation{%
  \institution{Kansas State University}
  \city{Manhattan}
  \state{Kansas}
  \country{United States}}

\author{Mohammad Akbar Bin Shah}
\email{mohammadakbar@ksu.edu}
\affiliation{%
  \institution{Kansas State University}
  \city{Manhattan}
  \state{Kansas}
  \country{United States}}

\author{Naimul Haque}
\email{naimul.haque@lsu.edu}
\affiliation{%
  \institution{Louisiana State University}
  \city{Baton Rouge}
  \state{Louisiana}
  \country{United States}}

\author{Shuangqing Wei}
\email{swei@lsu.edu}
\affiliation{%
  \institution{Louisiana State University}
  \city{Baton Rouge}
  \state{Louisiana}
  \country{United States}}

\author{George T. Amariucai}
\email{amariucai@ksu.edu}
\affiliation{%
  \institution{Kansas State University}
  \city{Manhattan}
  \state{Kansas}
  \country{United States}}

\renewcommand{\shortauthors}{Zareian Jahromi et al.}

\begin{abstract}
Several methods for auditing privacy in embedding spaces report a number called
``epsilon.'' The common name is misleading: one number may be a heuristic
score, another may come from a valid population inequality but ignore sampling
uncertainty, and a third may be a confidence bound. This paper asks when such a
number is evidence about differential privacy. We study four approaches based
on Gaussian calibration, marginal kernel-density ratios, maximum mean
discrepancy (MMD), and classifier hypothesis tests. The first two are modeling
diagnostics. The MMD and classifier approaches use valid population lower
bounds, but only the classifier approach is given a separate test set and a
finite-sample confidence calculation. To see how these distinctions matter, we
build a synthetic benchmark in which the true privacy value is known. It
includes pure-DP Laplace mechanisms, approximate-DP Gaussian mechanisms, an
exact privacy null, two identity geometries, and three sample sizes. At the
null, the Gaussian and kernel-density diagnostics remain large. A direct
conversion of an empirical ROC curve often returns infinity even though the
AUC is near chance and no threshold separates the samples perfectly. The MMD
value increases as the distributions become easier to distinguish, but the
sample estimate is small relative to the known reference and is not a lower
confidence bound. By contrast, a classifier chosen on development data and evaluated
on untouched test data gives simultaneous lower confidence bounds under the
stated iid model. We also apply the methods to FaceFusion and InstantID. That
case study shows how the methods behave on face data but does not calibrate
them because
the generators have no known privacy value. The main lesson is that an
epsilon-like number is meaningful only together with its assumptions and its
finite-sample interpretation.
\end{abstract}

\keywords{differential privacy, privacy auditing, face recognition, maximum
mean discrepancy, hypothesis testing, image-to-image generation}

\maketitle

\section{Introduction}
\label{sec:intro}

Face generators can alter an image substantially while leaving enough identity
information for a recognition model to recover the person. This makes visual
inspection a poor privacy test. An auditor instead compares the distributions
of identity embeddings produced from two source identities.

The difficulty is not computing a measure of separation. It is deciding what
that measure says about differential privacy (DP). Many embedding-space
procedures return a scalar denoted by $\varepsilon$, but the notation can hide
important differences. The value may depend on a Gaussian model, on an
independence approximation, or on a classifier and threshold chosen from the
same data used for evaluation. In these cases, a large value may indicate that
two samples differ, yet it need not be a statistically valid lower bound on the
privacy parameter. Related problems appear in feature-space privacy mechanisms
for face images~\cite{xue2023dpimage,wen2021identitydp} and in empirical
unlinkability evaluations for biometric templates~\cite{gomezbarrero2018unlinkability}.

We study four ways to turn embedding separation into an epsilon-like number.
\textsc{GaussMech} compares the distance between two sample means with their
within-identity variation. \textsc{KDE-LR} fits each embedding coordinate
separately and adds the largest marginal log ratios. We introduce these two as
diagnostics, not as established DP auditors. \textsc{MMD-TV} uses MMD and total
variation to obtain a population lower bound for pure DP. \textsc{ROC-HT} uses
the hypothesis-testing characterization of DP and a learned classifier. These
four approaches examine the same data from different angles: mean separation,
marginal density differences, joint distribution differences, and learned
classification. Their mathematical ingredients are standard
~\cite{dwork2014algorithmic,kairouz2015composition,gretton2012kernel,
sriperumbudur2010hilbert}, but their reported values do not have the same
statistical meaning.

Our question is simple: \emph{when does a finite-sample identity audit provide
evidence about $(\varepsilon,\delta)$-DP?} FaceFusion~\cite{facefusion} and
InstantID~\cite{wang2024instantid} cannot answer that question by themselves.
Both systems are designed to transfer or preserve identity, and neither claims
the identity-level DP guarantee studied here. More importantly, their true
privacy values are unknown. They are useful examples, but they cannot tell us
whether an audit is correctly calibrated.

We therefore begin with mechanisms whose privacy is known exactly. A finite
prototype set maps 400 identity labels to synthetic 512-dimensional vectors.
Laplace noise is calibrated using the exact global L1 sensitivity, and Gaussian
noise is calibrated using the analytic Gaussian profile and the exact global L2
sensitivity. A separate null mechanism gives every identity the same output
distribution. We repeat the experiment at three sample sizes and with two
prototype arrangements. Because the benchmark contains no faces and its pairwise
privacy values are known before the audits run, we can compare each reported
number with the quantity it is supposed to describe.

The benchmark exposes three problems. First, \textsc{GaussMech} and
\textsc{KDE-LR} report large values even when the two distributions are exactly
identical. Second, a direct conversion of an empirical ROC curve often returns
$+\infty$ at the null. This happens because a finite ROC curve can contain a
zero denominator even when its AUC is near $0.5$ and no threshold achieves
perfect separation. Third, the sample version of \textsc{MMD-TV} generally
increases with distinguishability, but it remains small relative to the known
privacy value and can exceed the truth in some low-privacy-loss settings. It is
therefore a useful descriptive statistic here, not a lower confidence bound.

The classifier method can be made statistically valid by separating selection
from evaluation. We choose the identity pair, classifier, orientation, and
threshold using development data. We then evaluate that fixed decision rule on
untouched test data and place simultaneous one-sided confidence bounds on its
two error rates. This produces a lower confidence bound on $\varepsilon$ under
the stated iid sampling model. The bound is valid by construction, although it
is often zero when the test set is small or the classifier is weak.

The synthetic benchmark does not show that a real image generator is private.
Its role is narrower: it reveals estimation artifacts and tests whether each
statistical interpretation holds when the answer is known. A real audit must
define the released channel, specify which inputs are adjacent, choose all
decision rules before looking at test data, and evaluate them on samples from
that same channel. Our FaceFusion and InstantID experiments illustrate this
process at the embedding level, subject to their stated data and provenance
limitations.

Previous DP auditors search for counterexamples, learn distinguishers, or
reduce the number of mechanism calls~\cite{ding2018statdp,bichsel2021dpsniper,
nasr2023tight,kong2024dpauditorium,lokna2023groupattack,
gonzalez2025sequential,xiang2025bits}. We do not propose a replacement auditor
or a new confidence theorem. Instead, we use known-privacy controls to show how
four common styles of analysis behave, then apply the same distinctions to face
generators.

\textbf{Contributions.}
\nopagebreak
\begin{itemize}[topsep=0pt,partopsep=0pt]
\item We distinguish three kinds of audit output: modeling diagnostics,
sample estimates inserted into valid population bounds, and finite-sample confidence
bounds for an identity-indexed channel.
\item We build a benchmark with known pairwise privacy values. It covers four
mechanism families, two synthetic identity geometries, three sample sizes, and
settings from an exact null to almost perfect distinguishability.
\item We identify specific finite-sample failures: the $\sqrt{d/n}$ null scale
of \textsc{GaussMech}, large null values from \textsc{KDE-LR}, small and
uncertified \textsc{MMD-TV} sample values, and infinite empirical ROC
conversions caused by endpoints rather than perfect separation.
\item We adapt a development/test split for simultaneous identity-pair bounds,
following classifier-based DP auditors~\cite{bichsel2021dpsniper,
kong2024dpauditorium}. We measure how often those bounds are informative and
apply the procedure to FaceFusion and InstantID only as a conditional
retrospective case study.
\end{itemize}

\section{What Should an Audit Report?}
\label{sec:problem}

\noindent\textbf{Three kinds of answers.}
Suppose we have samples from the outputs produced by two identities. An audit
can answer three different questions. A \emph{diagnostic} asks whether the
samples look different under a particular model. A \emph{population bound}
relates the true output distributions to DP, but replacing the unknown
distributions with samples does not automatically preserve that bound. A
\emph{finite-sample certificate} includes the sampling uncertainty and gives a
lower confidence bound on $\varepsilon$ for a stated threat model. These answers
should not all be reported as estimates of the same quantity. Table~\ref{tab:status}
shows which answer each method provides in this paper.

\begin{table*}[t]
\centering
\small
\caption{What each method reports. A ``No'' in the last column means that the
reported sample value is not a confidence bound on $\varepsilon$. The reserved-test
result is a confidence bound only under the iid model stated in
Section~\ref{sec:background}.}
\label{tab:status}
\begin{tabularx}{\textwidth}{@{}l>{\raggedright\arraybackslash}X
>{\raggedright\arraybackslash}X>{\raggedright\arraybackslash}Xl@{}}
\toprule
Method & What it measures & Assumption or requirement & Relation to DP &
Finite-sample confidence bound?\\
\midrule
\textsc{GaussMech} & mean-shift/residual score &
additive isotropic Gaussian model & only with known sensitivity/noise & No\\
\textsc{KDE-LR} & sum of marginal log ratios &
product/conditional factorization & not for the joint release here & No\\
\textsc{MMD-TV} & MMD-to-TV pure-DP value &
pure DP; $k(x,x)\leq1$ & valid population lower-bound formula & No\\
\textsc{ROC-HT} (OOF) & empirical ROC/AUC &
learned cross-validated test & valid for a fixed population test & No\\
\textsc{ROC-HT} (reserved test) & confidence-adjusted ROC bound &
fixed test; iid untouched data & valid hypothesis-testing route & Yes\\
\bottomrule
\end{tabularx}
\end{table*}

\begin{figure*}[t]
\centering
\includegraphics[width=0.9\textwidth]{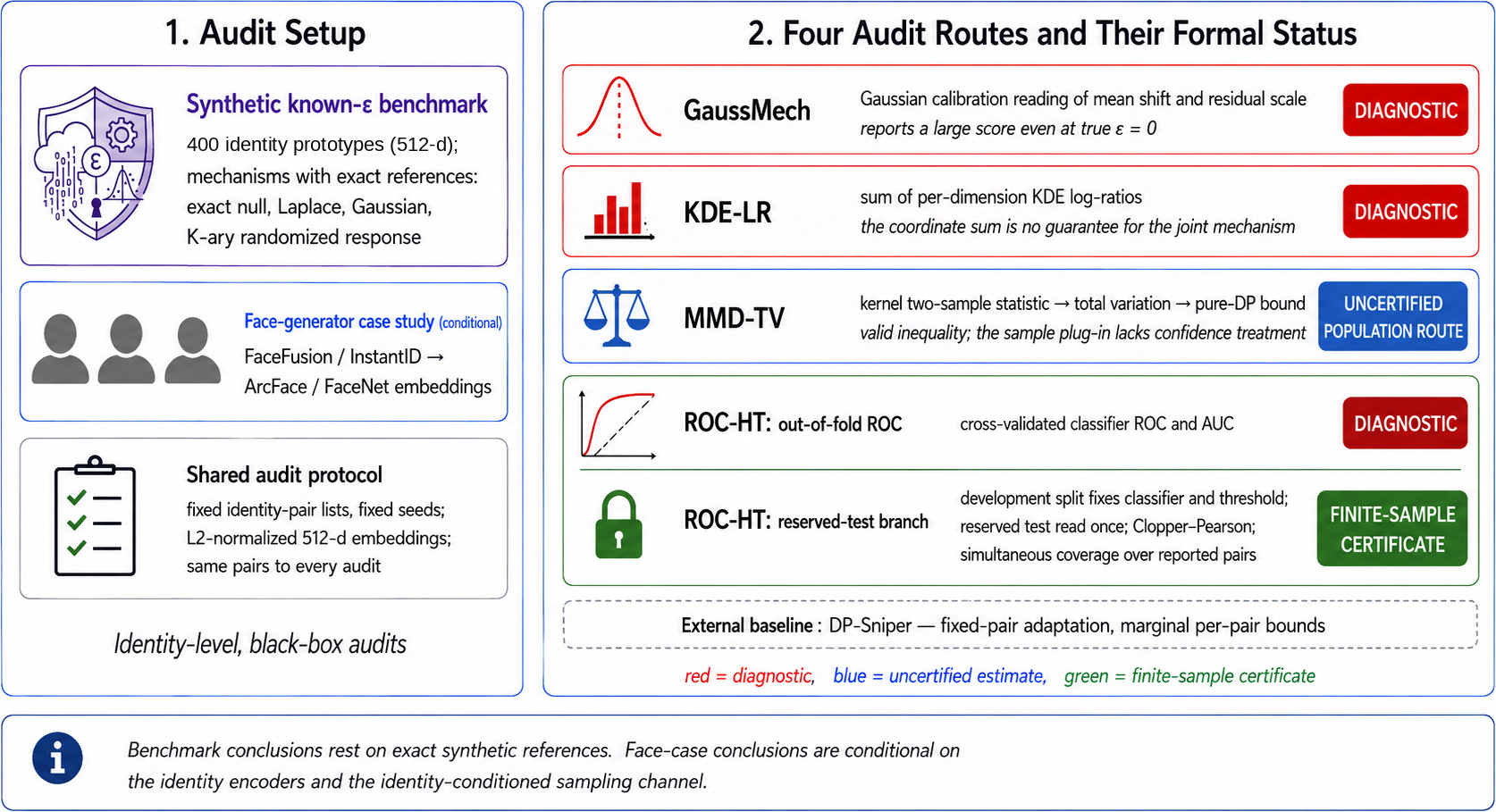}
\caption{The four audit routes use the same identity pairs but have different
formal status. \textsc{GaussMech} and \textsc{KDE-LR} are diagnostics.
\textsc{MMD-TV} inserts an uncertified sample estimate into a valid population
bound, and out-of-fold \textsc{ROC-HT} is descriptive. Only the separate ROC
test, chosen on development data and evaluated on reserved data, is a
finite-sample certificate.}
\Description{A two-panel overview. The left panel shows the audit setup: a
synthetic known-epsilon benchmark, a face-generator experiment with stated
assumptions,
and a shared audit protocol with fixed pairs and seeds. The right panel lists
the four audit routes and states what conclusion each result supports:
GaussMech and KDE-LR are diagnostic scores, MMD-TV is an uncertified
population route, and only the ROC-HT test on unused data gives a finite-sample
certificate. DP-Sniper appears as a separate comparison.}
\label{fig:pipeline}
\end{figure*}

\noindent\textbf{How we judge the methods.}
We ask four questions. Does the method return zero when the two population
distributions are identical? Does its output change when the known pairwise
$\varepsilon$ changes? When a method gives a confidence bound, how often is that
bound greater than zero? Finally, does the construction have the advertised
simultaneous confidence level? The third question measures usefulness, while
the fourth measures validity. A bound can be valid yet usually equal to zero.
We never rank methods by comparing raw numbers from different rows of
Table~\ref{tab:status}.

Every reported value keeps one of the three status labels in
Table~\ref{tab:status}. For a diagnostic or an uncertified population route,
we compare its behavior with matched mechanisms of known privacy, including an
exact null. For a certificate, we separately report how well the classifier
learns the distinguishing rule, the fraction of predeclared pairs whose lower
bound is positive, and the loss of power caused by the simultaneous confidence
correction. This separation prevents a large diagnostic score from being
mistaken for a certified privacy lower bound.

\noindent\textbf{How to test a privacy claim.}
Consider a proposed single-release guarantee $(\varepsilon_0,\delta_0)$. The
operator first specifies the identity populations, adjacency relation, and
confidence level. Development data may be used to search for identity pairs,
train a classifier, and choose its orientation and threshold. After these
choices are fixed, the test is run once on reserved data. If the simultaneous
lower bound $\varepsilon_L$ exceeds $\varepsilon_0$, the proposed guarantee is
false under the stated sampling model. Otherwise, the audit is inconclusive;
it has not proved the mechanism private and cannot by itself approve
deployment. A claim about repeated releases also needs a composition analysis.

\section{Background and Threat Model}
\label{sec:background}

\subsection{Differential Privacy}
\label{sec:bg-dp}

A randomized mechanism $\mathcal{M}: \mathcal{D} \to \mathcal{Y}$ is
$(\varepsilon, \delta)$-differentially private with respect to an adjacency
relation $\sim$ if, for every adjacent pair $D_0 \sim D_1$ and every measurable
$S \subseteq \mathcal{Y}$,
\begin{equation}
\Pr[\mathcal{M}(D_0) \in S] \leq e^{\varepsilon}\, \Pr[\mathcal{M}(D_1) \in S] + \delta. \label{eq:dp-def}
\end{equation}
Smaller $\varepsilon$ means stronger privacy. Setting $\delta=0$ gives pure
$\varepsilon$-DP~\cite{dwork2006calibrating}. In this paper, the input domain
$\mathcal{D}$ is a finite set of identity labels. Every two different labels
are adjacent, so the inequality must hold for every identity pair. Adjacency
here has nothing to do with visual similarity. Given an identity label,
$\mathcal{M}$ draws an image associated with that identity and releases one
transformed output (Section~\ref{sec:bg-threat}).

The same requirement can be expressed as a limit on classification accuracy.
Suppose a test tries to decide whether an output came from
$\mathcal{M}(D_0)$ or $\mathcal{M}(D_1)$. Let $\mathrm{FPR}$ be its false-positive
rate and $1-\mathrm{TPR}$ its false-negative rate. The mechanism is
$(\varepsilon,\delta)$-DP if and only if every test and every adjacent pair
satisfy
\begin{align}
\mathrm{FPR} + e^{\varepsilon}(1 - \mathrm{TPR}) &\geq 1 - \delta, \notag\\
(1 - \mathrm{TPR}) + e^{\varepsilon}\mathrm{FPR} &\geq 1 - \delta.
\label{eq:hyptest-dp}
\end{align}
Kairouz et al.~\cite{kairouz2015composition} establish this characterization,
and Dong et al.~\cite{dong2022gaussian} express it through $f$-DP. Our
\textsc{ROC-HT} method uses these inequalities directly.

\subsection{Total Variation and Differential Privacy}
\label{sec:bg-tv}

The MMD method needs a bridge from distributional distance to DP. Total
variation provides that bridge. Pure DP limits how far apart two adjacent
output distributions can be in total variation. The converse is not true: a
small total-variation distance alone does not prove DP.

For an adjacent pair $D_0\sim D_1$, let
$P_b=\mathrm{Law}(\mathcal{M}(D_b))$ for $b\in\{0,1\}$. A test predicts $D_1$
when the output falls in a set $A$. Its false-positive and false-negative rates
are $\beta_{\mathrm{I}}(A)=P_0(A)$ and
$\beta_{\mathrm{II}}(A)=1-P_1(A)$. The smallest possible sum of these errors
determines total variation:
\begin{equation}
d_{\mathrm{TV}}(P_0, P_1) = 1 - \inf_A \bigl[\beta_{\mathrm{I}}(A) + \beta_{\mathrm{II}}(A)\bigr]. \label{eq:tv-via-errors}
\end{equation}
Ghazi and Issa~\cite{ghazi2024totalvariation} (Section II, Definition 4 and
Theorem 1), building on Kairouz, Oh, and
Viswanath~\cite{kairouz2015composition}, show that every
$(\varepsilon,\delta)$-DP mechanism satisfies
$\beta_{\mathrm{I}}(A)+\beta_{\mathrm{II}}(A)
\geq 2(1-\delta)/(1+e^{\varepsilon})$ for every test. Combining this result
with Equation~\ref{eq:tv-via-errors} gives
\begin{equation}
d_{\mathrm{TV}}(\mathcal{M}(D_0), \mathcal{M}(D_1)) \;\leq\; \frac{e^{\varepsilon} - 1 + 2\delta}{e^{\varepsilon} + 1}. \label{eq:tv-dp}
\end{equation}
For pure DP, $\delta=0$, this becomes
$d_{\mathrm{TV}}\leq\tanh(\varepsilon/2)$. Solving for $\varepsilon$ gives
\begin{equation}
\varepsilon \;\geq\; \ln\!\left(\frac{1 + d_{\mathrm{TV}}}{1 - d_{\mathrm{TV}}}\right). \label{eq:tv-to-eps}
\end{equation}
This is the last step used by \textsc{MMD-TV}; Theorem~\ref{thm:mmdtv}
connects MMD to total variation.

\subsection{Face Identity Encoders and Generators}
\label{sec:bg-face}

A face encoder maps an image $x$ to a vector $\phi(x)\in\mathbb{R}^d$ that is
intended to represent identity. We use two sets of precomputed 512-dimensional
embeddings. One was produced by ArcFace through InsightFace
~\cite{deng2019arcface,insightface}. The other was produced by the
InceptionResnet-V1 model in \texttt{facenet-pytorch}~\cite{esler2019facenet}.
The exact checkpoint identifiers were not recorded,
so we do not infer checkpoint variants or training sets from the arrays. We
L2-normalize the face embeddings before analysis. The synthetic benchmark is
different: its prototypes begin on the unit sphere, but we do not normalize
the outputs after adding noise. In both settings, the encoder is a measurement
tool rather than ground truth. Any encoder error affects all four audits
(Section~\ref{sec:limitations}).

We study two generators. \textbf{FaceFusion}~\cite{facefusion} swaps a donor
face into a target image. We hold the donor fixed in the original eight
conditions, so each output label refers to the target identity.
\textbf{InstantID}~\cite{wang2024instantid} generates an image from a reference
identity, and each label refers to that reference identity. The official
InstantID pipeline uses InsightFace's \texttt{antelopev2} model
~\cite{instantidrepo}. If our precomputed outputs used that default, an
ArcFace/InsightFace audit would partly reuse the representation that conditioned
the generator. The exact configuration was not recorded. We therefore also
report FaceNet results, which provide a cleaner comparison across encoder
families. Neither generator is treated as a privacy mechanism.

\subsection{Maximum Mean Discrepancy}
\label{sec:bg-mmd}

MMD measures whether two distributions differ after mapping their samples into
a reproducing kernel Hilbert space. Let
$k:\mathcal{X}\times\mathcal{X}\to\mathbb{R}$ be a positive-definite kernel
with RKHS $\mathcal{F}_k$. For distributions $P_0$ and $P_1$, the maximum mean
discrepancy~\cite{gretton2012kernel} is
\begin{equation}
\mathrm{MMD}(P_0, P_1; \mathcal{F}_k)
= \sup_{\substack{f \in \mathcal{F}_k\\ \|f\|_{\mathcal{F}_k} \leq 1}}
\bigl|\mathbb{E}_{P_0}[f] - \mathbb{E}_{P_1}[f]\bigr|. \label{eq:mmd-def}
\end{equation}
For a characteristic kernel, MMD is zero exactly when the distributions are
equal. The RBF kernel
$k(u,v)=\exp(-\|u-v\|^2/(2\sigma^2))$ has this property for every
$\sigma>0$. Given samples $\{x_i\}_{i=1}^n\sim P_0$ and
$\{y_j\}_{j=1}^m\sim P_1$, we estimate squared MMD by
\begin{align}
\widehat{\mathrm{MMD}}^2_u
&= \frac{1}{n(n-1)} \sum_{i \neq i'} k(x_i, x_{i'}) \notag\\
&\quad + \frac{1}{m(m-1)} \sum_{j \neq j'} k(y_j, y_{j'}) \notag\\
&\quad - \frac{2}{nm} \sum_{i, j} k(x_i, y_j). \label{eq:mmd-unbiased}
\end{align}
This estimator is unbiased for the population $\mathrm{MMD}^2$ and supports a
permutation test of $H_0:P_0=P_1$~\cite{gretton2012kernel}. Later we substitute
a sample MMD into a bound that is valid for the population MMD. That
substitution produces a point estimate. It becomes a lower confidence bound
only if sampling uncertainty is handled separately.

\subsection{Threat Model}
\label{sec:bg-threat}

\noindent\textbf{One release.}
The curator is given one of two candidate identities, chosen by a hidden bit
$b$. It draws an image $X$ from that identity, transforms the image, and
releases $Y$. The auditor knows the procedure and the two candidates but does
not observe $X$ or $b$. It encodes $Y$ with the fixed face encoder $\phi$ from
Section~\ref{sec:bg-face} and tries to predict $b$. This is the privacy claim
that we audit. FaceFusion and InstantID themselves do not make this claim.

\noindent\textbf{What is random?}
Let $P_{\mathrm{id}}$ denote the distribution of images labeled with identity
$\mathrm{id}$. The variable $R$ collects any randomness inside the generator.
If a generator is deterministic once its inputs are fixed, then $R$ is
constant. For InstantID, $X\sim P_{\mathrm{id}}$ is the reference image and
$Y=g_{\mathrm{ID}}(X;R)$. For FaceFusion, $X\sim P_{\mathrm{id}}$ is the target
image, one donor $d_0$ is used for every identity, and
$Y=g_{\mathrm{FF}}(X;d_0,R)$. This gives two channels: the released image and
the embedding that we actually analyze.
\begin{equation}
\begin{aligned}
\mathcal{M}_{\mathrm{img}}(\mathrm{id})
&=\mathrm{Law}\!\left(Y\mid X\sim P_{\mathrm{id}},\,R\sim P_R\right),\\
\mathcal{M}_\phi(\mathrm{id})
&=\mathrm{Law}\!\left(\phi(Y)\mid X\sim P_{\mathrm{id}},\,R\sim P_R\right).
\end{aligned}
\label{eq:distributional-mechanism}
\end{equation}
In the five-pass cascade experiment, FaceFusion is deterministic after the
source and donor images have been selected; we do not vary a generator seed.
The records for the original eight conditions do not show whether generator
randomness varied. In every condition, however, drawing $X$ from an identity's
image collection is part of the mechanism. Each CelebA identity contributes
26--35 usable outputs, and each VGGFace2 identity contributes 90--100.

This model describes one release drawn from the image population for an
identity. It does not describe a user who chooses one fixed photograph. That
choice would condition on a particular $X$ instead of averaging over
$X\sim P_{\mathrm{id}}$. The setup is an identity-indexed instance of
distribution privacy~\cite{chen2023distributionprivacy} and has a similar
organization of secrets and distributions to Pufferfish
~\cite{kifer2014pufferfish}. The term \emph{identity-indexed} names the protected
input; Definition~\ref{def:id-dp} uses the usual pairwise DP inequality.

For InstantID, the private input is the reference identity. For FaceFusion, it
is the target identity. Because the same donor $d_0$ is used for both classes
in every pairwise test, a classifier cannot distinguish the classes merely by
recognizing different donors. After averaging over $X$ and $R$,
$\mathcal{M}_\phi$ is simply a finite channel from an identity label to a
distribution of embeddings.

This is different from user-level DP in model training, where neighboring
datasets differ in all records contributed by one person and the released
object is a trained model~\cite{xu2023userlevel}. It also differs from methods
that perturb one particular image or representation before release
~\cite{xue2023dpimage,wen2021identitydp,shibata2024dpddpm}. Mathematically, our
finite all-pairs channel also has the form of local DP. We use different
terminology because a curator draws an image from an identity-conditioned
population before releasing the output. We are not proposing a new definition
of DP.

\noindent\textbf{What an embedding audit can establish.}
Any two distinct identities are adjacent. Since $\phi$ is deterministic, DP
for the generated images would imply the same DP guarantee for their
embeddings. Therefore, a statistically certified violation in embedding space
also refutes the proposed guarantee for the images. The reverse does not hold:
failure to find a violation after encoding does not prove that the images are
private. There is also a connection to per-image claims. If every release were
$(\varepsilon_0,\delta_0)$-DP under replacement of one fixed image, joint
convexity of hockey-stick divergence implies the same bound after averaging
over two identity-conditioned image distributions. Thus a certified violation
of the identity-indexed channel would also refute that per-image claim
~\cite{kawamoto2019local}.

\begin{definition}[Identity-indexed DP]
\label{def:id-dp}
The distributional channel $\mathcal{M}_\phi$ is
\emph{identity-indexed $(\varepsilon, \delta)$-DP} if, for every ordered pair of
distinct identities $(\mathrm{id}_A, \mathrm{id}_B)$ and every measurable
$S \subseteq \mathbb{R}^d$,
\begin{equation*}
\mathcal{M}_\phi(\mathrm{id}_A)(S)
\leq e^{\varepsilon}\mathcal{M}_\phi(\mathrm{id}_B)(S)+\delta.
\end{equation*}
\end{definition}
Equivalently, every classifier that receives one embedding must obey the error
tradeoff in Equation~\ref{eq:hyptest-dp}. DP does not require the two identities
to be indistinguishable. It limits how small the two error rates can be at the
same time. The stored outputs provide repeated samples for estimating this
single-release channel. We make no group-privacy claim for releasing the whole
batch.

\section{Four Audit Methods}
\label{sec:methods}

\subsection{\textsc{GaussMech}: Gaussian Signal-to-Noise Score}
\label{sec:gaussmech}

\textsc{GaussMech} compares the distance between two identities' sample means
with the variation observed within each identity. A larger distance and less
within-identity variation produce a larger score. The calculation comes from
the usual Gaussian-mechanism formula, but that formula assumes additive,
input-independent isotropic Gaussian noise and a known global sensitivity. It
is also commonly stated for $\varepsilon\in(0,1)$
~\cite{dwork2014algorithmic}. Our face generators do not meet these conditions,
so we use the result as a signal-to-noise score rather than an estimate of DP
$\varepsilon$.

\noindent\textbf{Calculation.} For identity $\mathrm{id}$, let
$z_{\mathrm{id},i}\in\mathbb{R}^d$ denote its $i$th normalized generated embedding, and define
\begin{align*}
\widehat{\mu}_{\mathrm{id}}&=\frac{1}{n_{\mathrm{id}}}
\sum_{i=1}^{n_{\mathrm{id}}}z_{\mathrm{id},i},\\
\widehat{\sigma}_{\mathrm{id}}&=
\sqrt{\frac{1}{n_{\mathrm{id}}d}\sum_{i=1}^{n_{\mathrm{id}}}
\|z_{\mathrm{id},i}-\widehat{\mu}_{\mathrm{id}}\|_2^2},\\
\widehat{\Delta}_2(A,B)&=\|\widehat{\mu}_A-\widehat{\mu}_B\|_2.
\end{align*}
Here $\widehat{\sigma}_{\mathrm{id}}$ is the root-mean-square within-identity
residual per embedding coordinate. The textbook sufficient Gaussian calibration
$\sigma\geq\Delta_2\sqrt{2\ln(1.25/\delta)}/\varepsilon$
can be algebraically inverted~\cite{dwork2014algorithmic}. Replacing its known
sensitivity and noise scale with the observed separation and residual scales
gives
\begin{equation}
\widehat{\varepsilon}_{\textsc{GaussMech}}(A, B) \;:=\; \frac{\widehat{\Delta}_2(A, B)\,\sqrt{2 \ln(1.25/\delta)}}{\max(\widehat{\sigma}_A, \widehat{\sigma}_B)}. \label{eq:gaussmech}
\end{equation}
Once the identity means and scales have been computed, each pair costs $O(d)$.
We use the larger residual scale, which gives the smaller of the two possible
scores, and report the median across pairs. This choice does not establish the
DP inequality in either direction.

The model is also a poor fit to the face embeddings. After L2 normalization,
the coordinate residuals are concentrated but not Gaussian
(Figure~\ref{fig:qqplots}, Appendix~\ref{app:robustness}). A Shapiro--Wilk test
~\cite{shapiro1965analysis} rejects Gaussianity at $p<10^{-3}$ for 509 of 512
coordinates in the ArcFace+VGGFace2+InstantID condition and for a majority of
coordinates in every face condition. One residual scale also hides differences
in variance across coordinates. Finally, the formula treats estimated means
and variances as exact. It does not separate mechanism variation from sampling
error, and it can miss leakage that lies in a direction with high overall
variation.

We set $\delta=10^{-5}$ and use the root-mean-square residual per coordinate as
the scale. Because this is not a valid inversion of the Gaussian mechanism for
these generators, that value of $\delta$ should not be compared with the
pure-DP formulas used by \textsc{MMD-TV} and \textsc{ROC-HT}.

\subsection{\textsc{KDE-LR}: Coordinate-Wise Density-Ratio Score}
\label{sec:kdelr}

\textsc{KDE-LR} examines each embedding coordinate separately. For every
coordinate, it estimates the two identities' densities and finds the largest
log-density ratio on a finite grid. It then adds these coordinate scores. This
addition resembles the basic DP composition rule, but the coordinates of one
embedding are not separate mechanisms. The final sum is therefore a
description of the coordinate marginals, not a DP bound for the full vector.

\noindent\textbf{Calculation.} For identity $\mathrm{id}\in\{A,B\}$ and
coordinate $d'$, let $\hat{p}^{(d')}_{\mathrm{id}}$ estimate the marginal law
of $\phi(g(X))_{d'}$ conditional on $X\sim\mathrm{id}$. We use a Gaussian
kernel with Silverman's rule-of-thumb bandwidth~\cite{silverman1986density}.
For smoothing weight $\alpha\in[0,1]$, define
\begin{equation}
\hat{p}^{(d')}_{\mathrm{id}, \alpha}(x) \;:=\; (1 - \alpha)\,\hat{p}^{(d')}_{\mathrm{id}}(x) + \alpha\,u^{(d')}(x), \label{eq:kdelr-smooth}
\end{equation}
where $u^{(d')}$ is uniform over a finite interval containing both identities' samples plus a three-bandwidth margin. Let $\mathcal{G}^{(d')}_{A,B}$ be the 201-point grid on that interval, let $\tau=10^{-12}$, and write $\tilde p=\max\{\hat p,\tau\}$. The implemented diagnostic is
\begin{align}
\widehat{\varepsilon}^{(d')}_\alpha(A, B) &:=
\max_{x\in\mathcal{G}^{(d')}_{A,B}}
\left|\log\tilde p^{(d')}_{A,\alpha}(x)
-\log\tilde p^{(d')}_{B,\alpha}(x)\right|, \notag\\
\widehat{\varepsilon}_{\textsc{KDE-LR}}(A, B) &:=
\sum_{d'=1}^{d}\widehat{\varepsilon}^{(d')}_\alpha(A, B). \label{eq:kdelr}
\end{align}
We fit separate Silverman bandwidths for $A$ and $B$; the larger one is used
only to extend the grid. We report the median pair score for $\alpha=0$ and
$\alpha=10^{-2}$.

Unlike \textsc{GaussMech}, this method can detect differences in marginal shape
as well as differences in means. It can therefore respond to multimodal or
heavy-tailed coordinates. The price is that it discards dependence between
coordinates. Two joint distributions can have similar marginals but different
dependence, or small differences can accumulate across 512 coordinates. This
is why the reported sums reach the thousands in Table~\ref{tab:comparative};
they should not be read as joint $\varepsilon$ values.

The numerical choices also matter. Increasing $\alpha$ from $0$ to $10^{-2}$
reduces the score by a factor of $4.8$ to $6.4$ across the eight face
conditions, and the method itself does not choose between these settings. The
finite grid and density floor avoid divergent ratios but directly affect the
large values at $\alpha=0$. Our implementation uses Gaussian kernels,
identity-specific Silverman bandwidths, a 201-point grid extending three
bandwidths beyond the samples, and a density floor of $\tau=10^{-12}$.

\subsection{\textsc{MMD-TV}: A Population Lower Bound}
\label{sec:mmdtv}

\textsc{MMD-TV} starts from a simple idea: if the output distributions for two
identities are far apart, then a small value of $\varepsilon$ cannot describe
both of them. MMD measures part of this separation. We first use the known
inequality $\mathrm{MMD}\leq2d_{\mathrm{TV}}$ for bounded-diagonal kernels
~\cite{sriperumbudur2010hilbert}, then convert total variation to a lower bound
on pure-DP $\varepsilon$~\cite{kairouz2015composition,
ghazi2024totalvariation}. The result is valid for population distributions.
The value computed from finite samples is only an estimate unless it also
accounts for sampling error.

\noindent\textbf{Population formula.} Let $P_A$ and $P_B$ denote the
distributions of $\mathcal{M}_\phi$ under identities $A$ and $B$. For a
positive-definite kernel $k$ satisfying $k(x,x)\leq1$ for every $x$, the
population \textsc{MMD-TV} bound is
\begin{equation}
\varepsilon_{\textsc{MMD-TV}} \;:=\; \ln\!\left(\frac{1 + \mathrm{MMD}(P_A, P_B)/2}{1 - \mathrm{MMD}(P_A, P_B)/2}\right). \label{eq:mmdtv}
\end{equation}
The right-hand side is defined as $+\infty$ when $\mathrm{MMD}=2$. Its population justification follows from Lemma~\ref{lem:mmdtv} and Theorem~\ref{thm:mmdtv}; empirical results later evaluate the expression using a sample MMD estimate without a one-sided confidence correction.
This substitution is often called a \emph{plug-in estimate}: the unknown
population MMD is replaced by its value estimated from the observed samples.

\begin{lemma}[$\mathrm{MMD} \leq 2\,d_{\mathrm{TV}}$ for bounded-diagonal kernels]
\label{lem:mmdtv}
Let $k: \mathcal{X} \times \mathcal{X} \to \mathbb{R}$ be a positive-definite kernel with $k(x, x) \leq 1$ for all $x \in \mathcal{X}$. For any two distributions $P_0$, $P_1$ on $\mathcal{X}$,
\begin{equation}
\mathrm{MMD}(P_0, P_1; \mathcal{F}_k) \;\leq\; 2\,d_{\mathrm{TV}}(P_0, P_1). \label{eq:mmd-le-tv}
\end{equation}
\end{lemma}
\begin{proof}
For any $f \in \mathcal{F}_k$ with $\|f\|_{\mathcal{F}_k} \leq 1$, the reproducing property and Cauchy--Schwarz give $|f(x)| = |\langle f, k(\cdot, x) \rangle_{\mathcal{F}_k}| \leq \|f\|_{\mathcal{F}_k} \sqrt{k(x, x)} \leq 1$. Therefore the unit ball of $\mathcal{F}_k$ is contained in the unit ball of $L^\infty$, and
\begin{align*}
\mathrm{MMD}(P_0,P_1;\mathcal{F}_k)
&\leq \sup_{\|f\|_\infty\leq1}
\bigl|\mathbb{E}_{P_0}f-\mathbb{E}_{P_1}f\bigr|\\
&=2\,d_{\mathrm{TV}}(P_0,P_1).
\end{align*}
This is a specialization of the MMD-vs.-IPM comparison in~\cite{sriperumbudur2010hilbert}.
\end{proof}

\begin{theorem}[MMD-TV analytical lower bound on $\varepsilon$]
\label{thm:mmdtv}
Let $\mathcal{M}$ be any pure $\varepsilon$-DP mechanism on an adjacent pair $D_0\sim D_1$. Let $P_b=\mathrm{Law}(\mathcal{M}(D_b))$ for $b\in\{0,1\}$, and let $k$ be a positive-definite kernel with $k(x, x) \leq 1$. Then
\begin{equation}
\varepsilon \;\geq\; \ln\!\left(\frac{1 + \mathrm{MMD}(P_0, P_1)/2}{1 - \mathrm{MMD}(P_0, P_1)/2}\right). \label{eq:mmdtv-thm}
\end{equation}
The bound is valid at the population level for any pure $\varepsilon$-DP
mechanism. Equality in the MMD-to-TV step would require the RKHS witness to
represent a TV-optimal sign function for the signed measure $P_0-P_1$, which is
generally unavailable. We restrict the \textnormal{\textsc{MMD-TV}} statement
to pure DP rather than using an approximate-DP slack term.
\end{theorem}
\begin{proof}
The pure-DP TV bound~\cite{kairouz2015composition,ghazi2024totalvariation} states that any pure $\varepsilon$-DP mechanism satisfies $d_{\mathrm{TV}}(P_0,P_1) \leq \tanh(\varepsilon/2) = (e^\varepsilon - 1)/(e^\varepsilon + 1)$. Inverting the bijection $t \mapsto \tanh(t/2)$ from $[0, \infty)$ to $[0, 1)$ gives $\varepsilon \geq 2\,\mathrm{arctanh}(d_{\mathrm{TV}}) = \ln((1 + d_{\mathrm{TV}})/(1 - d_{\mathrm{TV}}))$. By Lemma~\ref{lem:mmdtv}, $d_{\mathrm{TV}}(P_0,P_1) \geq \mathrm{MMD}(P_0,P_1)/2$. The map $t \mapsto \ln((1+t)/(1-t))$ is monotone increasing on $[0, 1)$, so substituting $\mathrm{MMD}(P_0,P_1)/2$ for $d_{\mathrm{TV}}(P_0,P_1)$ can only decrease the right-hand side, preserving the lower bound.
\end{proof}

\noindent\textbf{Relation to an existing tighter bound.} Gonz\'alez et al.
~\cite{gonzalez2025sequential}, improving Kong et al.
~\cite{kong2024dpauditorium}, show for $0\leq k(x,y)\leq1$ that
\begin{equation*}
\mathrm{MMD}(P_0,P_1)
\leq\sqrt{2}\left(1-\frac{2(1-\delta)}{1+e^\varepsilon}\right).
\end{equation*}
At $\delta=0$, this inverts to
$\varepsilon\geq2\,\mathrm{arctanh}(\mathrm{MMD}/\sqrt{2})$, which is tighter
than Equation~\ref{eq:mmdtv-thm} for the nonnegative RBF kernel used here. Our
derivation needs only positive definiteness and $k(x,x)\leq1$, and it shows
where the loss in the bound enters. We do not claim that it improves the
earlier result.

\noindent\textbf{What we compute from samples.} We use an RBF kernel whose
median-heuristic bandwidth is selected separately from each pair's pooled
samples. The implementation computes
\begin{equation*}
\widehat{\mathrm{MMD}}
=\sqrt{\max\{\widehat{\mathrm{MMD}}_u^2,0\}},
\end{equation*}
substitutes it into Equation~\ref{eq:mmdtv-thm}, and reports the largest value
among 500 audited pairs. Because the bandwidth is data-dependent and the
unbiased squared statistic is clipped and transformed, the complete statistic
is not an unbiased estimator of a fixed population MMD. A condition with 399
identities has 79,401 candidate pairs, while one with 400 has 79,800. The
reported maximum is neither the worst case over all pairs nor a lower
confidence bound for the population.

The population result does not assume Gaussian noise or independent
coordinates, and MMD compares the full embedding rather than one coordinate at
a time. Its
main weakness is numerical looseness. The classifier-based empirical TV score
is 4.4--7.5 times larger than the MMD/2 estimate on average. Because the
classifier may also fall short of total variation, this gap cannot be assigned
entirely to Lemma~\ref{lem:mmdtv}. Every face-data value is below one, so even a
confidence-corrected version at these magnitudes would not rule out
$\varepsilon=1$.

There is also a gap between the theorem and the reported statistic. Even with
a fixed bandwidth, the unbiased estimator of squared MMD can exceed its
population value. We additionally choose the bandwidth from the same data,
clip negative estimates to zero, and take a square root. The final number is
therefore a point estimate, not a high-probability lower bound. We use an RBF
kernel, the pairwise median bandwidth, and 500 pairs sampled without
replacement. Testing more pairs can find a more difficult pair, but it does
not improve the estimate for any one pair.

\subsection{\textsc{ROC-HT}: A Classifier-Based Privacy Test}
\label{sec:rocht}

\textsc{ROC-HT} asks how well a classifier can tell two identities apart from
one released embedding. Differential privacy limits the combinations of
false-positive and false-negative rates that any such classifier can achieve
~\cite{wasserman2010statistical,kairouz2015composition,dong2022gaussian}. If a
fixed classifier makes both errors rarely enough, those error rates give a
lower bound on $\varepsilon$. We use this idea in two ways. Cross-validated ROC
curves provide a broad descriptive comparison. A separate experiment chooses
the classifier and threshold on development data and evaluates them on unused
test data; only the latter includes a finite-sample confidence guarantee.

\noindent\textbf{Calculation.} For each identity pair $(A,B)$, fit a
cross-validated classifier $h_{A,B}:\mathbb{R}^d\to[0,1]$ to labeled embeddings
from $A$ and $B$. Concatenate the out-of-fold scores and form the per-pair ROC
$\{(\widehat{\mathrm{FPR}}_i,\widehat{\mathrm{TPR}}_i)\}_i$. Applying the
Kairouz--Oh--Viswanath inequalities at each operating point gives
\begin{equation}
\widehat{\varepsilon}_{\textsc{ROC-HT},\delta}(A,B)
:=\max_i\left\{0,\,
\ln\frac{1-\delta-\widehat{\mathrm{FPR}}_i}
{1-\widehat{\mathrm{TPR}}_i},\,
\ln\frac{\widehat{\mathrm{TPR}}_i-\delta}
{\widehat{\mathrm{FPR}}_i}\right\}.
\label{eq:rocht}
\end{equation}
A term with a nonpositive numerator is omitted; a positive numerator over a
zero denominator gives $+\infty$, while $0/0$ is omitted. We use $\delta=0$
for the Laplace benchmark and face-generator case study and $\delta=10^{-5}$
for the Gaussian benchmark. We use 5-fold stratified cross-validation with a
fixed random seed and record the per-pair empirical value of
Equation~\ref{eq:rocht} over the observed out-of-fold operating points. Because
the classifier and maximizing threshold are data-dependent, these values are
descriptive rather than confidence bounds.

This method does not assume a parametric output distribution, and a better
classifier can reveal a larger lower bound. Cross-validation reduces the
optimism caused by evaluating a classifier on its training examples, but it
does not remove all reuse of the data: five different classifiers produce the
pooled scores, and the best threshold is chosen from those same scores. We
therefore reserve the word ``certificate'' for the separate test in
Section~\ref{sec:heldout-cert}.

The direct empirical conversion also has an important endpoint problem. It is
$+\infty$ when the observed ROC contains a point with
$\widehat{\mathrm{TPR}}=1$ and $\widehat{\mathrm{FPR}}<1-\delta$, or with
$\widehat{\mathrm{FPR}}=0$ and $\widehat{\mathrm{TPR}}>\delta$. For
$\delta=0$, a finite sample can produce either endpoint even when no threshold
perfectly separates the two classes. We therefore report AUC and separately
count the exact corner
$(\widehat{\mathrm{FPR}},\widehat{\mathrm{TPR}})=(0,1)$. More samples may
reveal overlap and remove an endpoint that appears in a smaller sample.

The result also depends on the classifier. We use logistic regression and
report a random-forest comparison in Appendix~\ref{app:robustness}; neither is
treated as uniquely correct. We use five stratified folds. The available data
provide about 100 samples per identity for VGGFace2 and about 30 for CelebA.

\section{Comparative Audit of Face Generators}
\label{sec:experiments}

We next apply the four methods to FaceFusion and InstantID only as a
conditional retrospective case study. Because these generators have no known
identity-level privacy value, this experiment cannot calibrate or rank the
methods. It shows what each method reports on face data under the stated
sampling assumptions. The available records also omit some details of the
original generation runs; Appendix~\ref{app:limitations} lists the fields a
fully documented rerun would require. Section~\ref{sec:exp-main} keeps these
limitations separate from the observed distinguishability results.

\subsection{Setup}
\label{sec:exp-setup}

\noindent\textbf{Datasets.} We use selected identity subsets from
VGGFace2~\cite{cao2018vggface2} and CelebA~\cite{liu2015celeba}. Seven
conditions retain 400 identities; one condition, using FaceNet, CelebA, and
FaceFusion, retains 399 after filtering. VGGFace2 identities have 90--100
usable generated embeddings and CelebA identities have 26--35.

\noindent\textbf{Generated images.} Our input consists of eight CSV files,
one for each encoder, dataset, and generator combination. Each row contains an
identity label and the embedding of a generated image. In the FaceFusion
experiments~\cite{facefusion}, the same donor is used for all target
identities, so the label identifies the target. In the InstantID experiments
~\cite{wang2024instantid}, the label identifies the reference person. The
saved files do not record the exact software revisions, prompts, random seeds,
checkpoint hashes, or FaceFusion donor identifier. In the tables, FF and ID
mean FaceFusion and InstantID; AF and FN mean ArcFace and FaceNet; CB and VGG
mean CelebA and VGGFace2.

\noindent\textbf{Face encoders.} One set of files contains ArcFace embeddings
produced through InsightFace; the other contains embeddings from
\texttt{facenet-pytorch}. The exact checkpoint names were not saved. We remove
a row if its generated embedding is missing or has
norm at most $10^{-6}$, then L2-normalize the remaining vectors. The two
encoders can lose different rows during this step, so their results are not
based on exactly the same images. They instead show whether the main finding
persists with a different representation. This comparison is especially
useful for InstantID, whose official configuration uses InsightFace
\texttt{antelopev2}~\cite{instantidrepo}. FaceNet avoids direct reuse of that
encoder family. However, both public face encoders were trained on large
web-scraped datasets whose identities may overlap the evaluation data or each
other. Agreement between them is therefore a robustness check, not two
statistically independent audits.

\noindent\textbf{Identity pairs.} The 399- and 400-identity conditions contain
79,401 and 79,800 possible unordered pairs. A fixed seed shuffles each list.
All four methods use the first 100 pairs, which makes the medians in
Table~\ref{tab:comparative} directly comparable as summaries of the same data.
For additional checks, \textsc{GaussMech} and \textsc{MMD-TV} use pairs
101--500, and \textsc{ROC-HT} uses pairs 101--200. Appendix
~\ref{app:route-config} gives the numerical settings.

\noindent\textbf{Saved records.} The audit script reads the saved embedding
files; it does not regenerate images or recompute embeddings. The artifact
records file hashes, row and identity counts, the software environment, and the
shared identity-pair lists. It also states which generation details are
unavailable.

\subsection{Results on Eight Face Conditions}
\label{sec:exp-main}

Table~\ref{tab:comparative} summarizes the same first 100 identity pairs for
each condition. The columns should not be compared as four estimates of one
quantity: \textsc{GaussMech} and \textsc{KDE-LR} are diagnostic scores,
\textsc{MMD-TV} is a sample estimate inserted into a population formula, and
the out-of-fold \textsc{ROC-HT} value is descriptive. Table
~\ref{tab:heldout-cert} gives the separate ROC confidence bounds obtained from
unused test data under an iid row model.

\begin{table*}[t]
\caption{Results on L2-normalized face embeddings. Finite values are medians
over the same 100 identity pairs. The \textsc{ROC-HT} column counts empirical
out-of-fold estimates equal to $+\infty$ because of a zero denominator; it does
not claim that the population value is infinite. Each method uses the
$\delta$ stated in Section~\ref{sec:methods}.
$^{\dagger}$A separate fixed-donor control on original VGGFace2 images,
without face restoration, whose generation settings were recorded. It uses
its own identity-pair list and is included only as a diagnostic condition.}
\label{tab:comparative}
\centering
\scriptsize
\setlength{\tabcolsep}{3.0pt}
\begin{tabular}{@{}lllrrrrrr@{}}
\toprule
Encoder & Dataset & Gen. & \shortstack{\textsc{GaussMech}\\diagnostic} & \shortstack{\textsc{KDE-LR} diagnostic\\$\alpha{=}0$} & \shortstack{\textsc{KDE-LR} diagnostic\\$\alpha{=}10^{-2}$} & \shortstack{\textsc{MMD-TV}\\sample est.} & \shortstack{\textsc{ROC-HT} OOF\\$\widehat{\varepsilon}=+\infty$} & \shortstack{ROC\\AUC} \\
\midrule
ArcFace & VGGFace2 & FF & $93.1$ & $10{,}447$ & $1{,}657$ & $0.475$ & $100/100$ & $1.000$ \\
ArcFace & VGGFace2 & ID & $46.1$ & $8{,}560$  & $1{,}334$ & $0.236$ & $100/100$ & $0.982$ \\
ArcFace & CelebA   & FF & $81.6$ & $8{,}271$  & $1{,}689$ & $0.388$ & $100/100$ & $1.000$ \\
ArcFace & CelebA   & ID & $58.7$ & $7{,}249$  & $1{,}523$ & $0.257$ & $100/100$ & $0.978$ \\
\midrule
FaceNet & VGGFace2 & FF & $91.8$ & $9{,}976$ & $1{,}652$ & $0.459$ & $100/100$ & $0.999$ \\
FaceNet & VGGFace2 & ID & $44.2$ & $7{,}872$ & $1{,}278$ & $0.217$ & $100/100$ & $0.937$ \\
FaceNet & CelebA   & FF & $80.3$ & $8{,}659$ & $1{,}721$ & $0.381$ & $100/100$ & $0.992$ \\
FaceNet & CelebA   & ID & $54.0$ & $7{,}190$ & $1{,}502$ & $0.243$ & $100/100$ & $0.921$ \\
\addlinespace[2pt]
ArcFace & VGG2-orig$^{\dagger}$ & FF & $107.9$ & $10{,}500$ & $1{,}717$ & $0.525$ & $100/100$ & $1.000$ \\
\bottomrule
\end{tabular}
\end{table*}

The face identities are readily distinguishable. Across the eight conditions,
the median AUC ranges from 0.921 to 1.000. Median \textsc{MMD-TV} ranges from
0.217 to 0.475, and the separate VGGFace2 condition reaches 0.525 on its own
pair list. For every condition, the direct ROC conversion is $+\infty$ for all
200 tested pairs because each finite ROC curve contains a zero-denominator
endpoint. Yet the exact $(0,1)$ corner occurs for only 0 to 130 of those 200
pairs, or 147 in the separate VGGFace2 condition. Thus the infinite value is
caused by the conversion at an endpoint, not by evidence of perfect population
separation. The largest \textsc{MMD-TV} estimate among 500 pairs ranges from
0.406 to 0.719, and is 0.724 for the separate condition.
Figure~\ref{fig:mmdroc} in Appendix~\ref{app:robustness} shows the MMD
distributions and one out-of-fold ROC curve for each condition.

We also repeat FaceFusion five times in succession. Each pass uses a newly
sampled donor, and donors are sampled in the same way for every target
identity. Observed classifier separation weakens after each pass, but the
reserved test still finds leakage after the fifth pass. All 50 identity pairs
chosen on development data have positive lower bounds on the unused test data;
46 bounds exceed one, and the median is $1.701$ at $\delta=10^{-5}$ with 95\%
simultaneous confidence. On a separate set of 200 pairs, the direct ROC
conversion remains $+\infty$ for every pair. These lower bounds do not show
that the true $\varepsilon$ decreased, or by how much. Tables
~\ref{tab:random-cascade-cert} and~\ref{tab:random-cascade-diag} give the full
results.

\subsection{Testing Fixed Classifiers on Unused Face Data}
\label{sec:heldout-cert}

The out-of-fold results above do not account for all sampling uncertainty. To
obtain confidence bounds, we make every choice before examining the test data.
For each identity, we assign 70\% of the samples to development and keep 30\%
unused for testing. Using development data alone, we score the first 100 pairs,
select the five most distinguishable pairs without reusing an identity, and
choose a classifier orientation and threshold for each pair. We then fit one
logistic-regression classifier per selected pair on all of its development
data. Finally, each fixed classifier and threshold is evaluated once on the
unused samples. Avoiding repeated identities keeps a few people from supplying
most of the evidence, although the confidence calculation does not require
that restriction. This is the usual train-then-test pattern used by
classifier-based auditors such as DP-Sniper and DP-Auditorium
~\cite{bichsel2021dpsniper,kong2024dpauditorium}; we do not claim a new
confidence theorem.

For each pair, let $U_{\mathrm{FPR}}$ and $U_{\mathrm{FNR}}$ be one-sided
Clopper--Pearson upper confidence bounds~\cite{clopper1934confidence} for the
two test error rates. There are two error rates for each of five pairs in each
of eight conditions, so we use
$\alpha=0.05/(8\cdot5\cdot2)=0.000625$ for each rate. Under the iid row model,
a union bound makes all 40 pairwise results correct together with probability
at least 95\%; it does not require the pairs or conditions to be independent.
Substituting the two error-rate limits into Equation~\ref{eq:hyptest-dp} gives
\begin{equation}
\varepsilon_{\mathrm{LB}}(\delta)
=\max\!\left\{0,\,
\ln\frac{1-\delta-U_{\mathrm{FNR}}}{U_{\mathrm{FPR}}},\,
\ln\frac{1-\delta-U_{\mathrm{FPR}}}{U_{\mathrm{FNR}}}
\right\},
\label{eq:heldout-eps}
\end{equation}
where a term is omitted if its numerator is nonpositive. Choosing pairs and
thresholds on the development data does not affect this calculation because
the test data remain unused until the final evaluation.

If the saved rows are iid samples from each $P_{\mathrm{id}}$, the 40 bounds
in Table~\ref{tab:heldout-cert} hold together with at least 95\% confidence.
Under that assumption, all 20 VGGFace2 bounds are positive and 15 exceed one.
FaceNet+VGGFace2+InstantID gives the weakest bounds, $0.685$--$0.925$. Because
the original generation records are incomplete, we cannot verify the iid
assumption. The numbers are therefore binomial confidence bounds under that
assumption, not unconditional guarantees about the generators.

CelebA shows the effect of a small test set. Eleven of 20 selected pairs make
no test error, yet with only 9--10 test samples per class and the 40-pair
correction, the
error upper bound is at least $0.5218$ (and $0.5595$ where a class has 9) and Equation~\ref{eq:heldout-eps} is zero. A balanced
zero-error test gives $\varepsilon_{\mathrm{LB}}=0$, $0.454$, and $1.277$ at
$n=10,15,30$ per class. These adaptively selected pairs are stress tests, not a
prevalence estimate or mechanism-wide supremum; OOF and MMD-TV remain
uncertified. Appendix~\ref{app:robustness} reports additional controls. With no
privacy ground truth, the face case cannot calibrate or rank the methods.

\begin{table}[t]
\caption{Results from unused test data for five selected identity pairs per
condition. Under the iid row model, all 40 bounds hold together with at least
95\% confidence. The saved generation records cannot verify that model; see
Appendix~\ref{app:limitations}. ``Pos.'' and ``$>1$'' count bounds above zero
and above one.}
\label{tab:heldout-cert}
\centering
\scriptsize
\setlength{\tabcolsep}{2.5pt}
\begin{tabular}{lccccc}
\toprule
Condition & $n$/class & AUC & Pos. & $>1$ & $\varepsilon_{\mathrm{LB}}(0)$ \\
\midrule
AF/VGG/FF & $30$ & $1.000$--$1.000$ & $5/5$ & $5/5$ & $1.192$--$1.277$ \\
AF/VGG/ID & $29$--$30$ & $0.986$--$1.000$ & $5/5$ & $5/5$ & $1.151$--$1.277$ \\
AF/CB/FF & $9$--$10$ & $1.000$--$1.000$ & $0/5$ & $0/5$ & $0$ \\
AF/CB/ID & $9$ & $1.000$--$1.000$ & $0/5$ & $0/5$ & $0$ \\
\midrule
FN/VGG/FF & $30$ & $0.997$--$1.000$ & $5/5$ & $5/5$ & $1.192$--$1.277$ \\
FN/VGG/ID & $29$--$30$ & $0.970$--$0.999$ & $5/5$ & $0/5$ & $0.685$--$0.925$ \\
FN/CB/FF & $9$--$10$ & $1.000$--$1.000$ & $0/5$ & $0/5$ & $0$ \\
FN/CB/ID & $9$ & $0.914$--$1.000$ & $0/5$ & $0/5$ & $0$ \\
\bottomrule
\end{tabular}
\end{table}

\section{Experiments with Known \texorpdfstring{$\varepsilon$}{Epsilon}}
\label{sec:known-eps}

\subsection{How We Generate Data with Known Privacy}

\noindent\textbf{Two sets of identity prototypes.}
For each identity $i$ in the finite input domain $\mathcal{D}$, let
$c_i\in\mathbb{R}^{512}$ be a fixed prototype. We construct two data-free sets
of 400 prototypes. In the dispersed set, we generate independent standard
Gaussian vectors with a fixed seed and normalize them to unit length. In the
clustered set, we first generate four normalized Gaussian centers. Identity
$i$ is assigned to center $i\bmod4$, and its prototype is that center plus
$0.18$ times an independently generated normalized direction, normalized
again. This second set lets us check whether nearby groups of identities make
the audit harder. Both sets and their random seeds are fixed before any data
are generated, and their array hashes are recorded in the artifact.

The exact $(\Delta_1,\Delta_2)$ values are $(28.302,1.543)$ for the dispersed
set and $(26.367,1.451)$ for the clustered set. We compute them by checking
every pair:
\begin{equation}
\Delta_p=\max_{i\ne j}\lVert c_i-c_j\rVert_p,\qquad p\in\{1,2\},
\label{eq:exact-sensitivity}
\end{equation}
The prototypes that attain each maximum are also recorded. Thus the
sensitivity is known exactly rather than bounded by the looser diameter
$2R$.

\noindent\textbf{Adding Laplace noise.}
The coordinatewise Laplace mechanism is
\begin{equation}
\mathcal{M}_{P}(i)=c_i+\eta,\qquad
\eta_\ell\stackrel{\mathrm{iid}}{\sim}\mathrm{Laplace}(0,b).
\end{equation}
For $\varepsilon_g>0$, setting $b=\Delta_1/\varepsilon_g$ gives global pure
$\varepsilon_g$-DP on the finite identity domain. The exact reference for pair
$(i,j)$ is
\begin{equation}
\varepsilon_{ij}^{P}=\lVert c_i-c_j\rVert_1/b.
\label{eq:laplace-pair-eps}
\end{equation}

\noindent\textbf{Adding Gaussian noise.}
The Gaussian mechanism is
\begin{equation}
\mathcal{M}_{G}(i)=c_i+\xi,\qquad
\xi\sim\mathcal{N}(0,\sigma^2I).
\end{equation}
For each target $\varepsilon_g$, we invert the analytic Gaussian privacy
profile~\cite{balle2018gaussian} using $\Delta_2$ and
$\delta=10^{-5}$. Replacing $\Delta_2$ by
$\lVert c_i-c_j\rVert_2$ in the same profile gives the exact pair-specific
reference $\varepsilon_{ij}^{G}$ at that $\delta$.
No post-noise normalization is applied; the audits observe the additive
mechanism outputs directly. The privacy references are exact for the ideal
real-valued mechanisms defined above. The numerical implementation generates
finite-precision Monte Carlo samples and is not claimed to provide bit-level
differential privacy.

\noindent\textbf{A case with no identity information.}
For target $\varepsilon_g=0$, every identity maps to the same prototype before
the common noise law is applied. Thus all identity-conditioned output laws are
identical and both global and pairwise $\varepsilon$ equal zero exactly.

\subsection{Experimental Design}

We evaluate the target global grid
$\{0,\allowbreak 0.1,\allowbreak 0.25,\allowbreak 0.5,\allowbreak 1,\allowbreak 2,\allowbreak 4,\allowbreak 8,\allowbreak 16,\allowbreak 32,\allowbreak 64,\allowbreak 128,\allowbreak 256\}$, sample sizes
$n\in\{30,100,300\}$ per identity, and three independent repetitions. Forty
identity pairs include the L1- and L2-sensitivity-attaining pairs and
distance-stratified pairs under both metrics. Every method receives the same
pair list. Across the two prototype sets and the Laplace and isotropic Gaussian
mechanisms, this gives 468 experimental settings and 18,720 pairwise results.

For the Gaussian family, the population-optimal AUC reference is
\begin{equation}
\mathrm{AUC}^{*}_{ij}
=\Phi\!\left(\frac{\lVert c_i-c_j\rVert_2}{\sqrt{2}\sigma}\right).
\label{eq:gaussian-bayes-auc}
\end{equation}
For product Laplace noise, we approximate the best possible population AUC
using the exact likelihood-ratio score and 20,000 Monte Carlo draws per class.
These reference values describe ideal classification performance, not
finite-sample confidence bounds. A measured AUC can fall above or below them
because of sampling variation. Their purpose is to separate information lost
by the privacy mechanism from information missed by our trained classifier.

\subsection{What Happens When \texorpdfstring{$\varepsilon=0$}{Epsilon = 0}?}

The zero-privacy-loss case is the clearest test because all identities have
exactly the same output distribution. As expected, Table~\ref{tab:null} shows
median AUC near $0.5$ and \textsc{MMD-TV} near zero. The other results reveal
finite-sample artifacts. \textsc{GaussMech} and \textsc{KDE-LR} remain large.
Across both prototype sets and both noise mechanisms, the direct out-of-fold
ROC conversion is $+\infty$ for 897 of 1,440 pairwise ROC curves (62.3\%), even
though none contains the exact $(0,1)$ corner. A finite ROC curve can have an
initial region with no observed errors, producing a zero denominator without
implying infinite population $\varepsilon$.

The clipped \textsc{MMD-TV} sample statistic is positive in 339 of 720
pure-DP null results (47.1\%) and reaches 0.057. A positive sample value alone
is therefore neither a rejection of equal distributions nor a confidence
bound above zero.

\begin{table}[t]
\centering
\small
\caption{Results when the true pairwise $\varepsilon$ is zero. For each
setting, we take the median over 40 pairs, or count the infinite ROC
conversions, and then average over both prototype sets, the applicable noise
mechanisms, and three repetitions. MMD-TV is shown only for the pure-DP Laplace
mechanism. ``Div.'' is the mean number of infinite ROC conversions among 40
pairs.}
\label{tab:null}
\begin{tabular}{rrrrrr}
\toprule
$n$ & AUC & Div. & MMD-TV & GaussMech & KDE-LR\\
&&&&& $\alpha=.01$\\
\midrule
30  & .490 & 29.3 & .003 & 28.60 & 1526\\
100 & .499 & 24.7 & .003 & 15.51 & 1210\\
300 & .502 & 20.8 & .001 &  8.96 &  962\\
\bottomrule
\end{tabular}
\end{table}

For an independent-coordinate Gaussian or Laplace null with common
per-coordinate standard deviation $\sigma$, the root second moment of the
distance between two empirical means has leading scale
$\sigma\sqrt{2d/n}$. Substitution into
Equation~\ref{eq:gaussmech} predicts the diagnostic floor
\begin{equation}
\sqrt{2\ln(1.25/\delta)}\sqrt{2d/n}.
\label{eq:gauss-null-floor}
\end{equation}
At $d=512$ and $\delta=10^{-5}$, this approximation predicts the three
\textsc{GaussMech} values in Table~\ref{tab:null}. It does not include the exact
expectation of the norm or the error in the estimated residual scales, but it
shows why the score remains positive: it is measuring error in the two sample
means even when the population means are equal. We also repeat the null
experiment for $d\in\{64,128,256,512\}$ and $n\in\{30,100,300\}$. Every
repetition agrees with Equation~\ref{eq:gauss-null-floor} to within 3.3\%,
supporting the scaling approximation without treating it as an exact formula.

\subsection{Do the Methods Track Known \texorpdfstring{$\varepsilon$}{Epsilon}?}

\begin{figure*}[t]
\centering
\includegraphics[width=\textwidth]{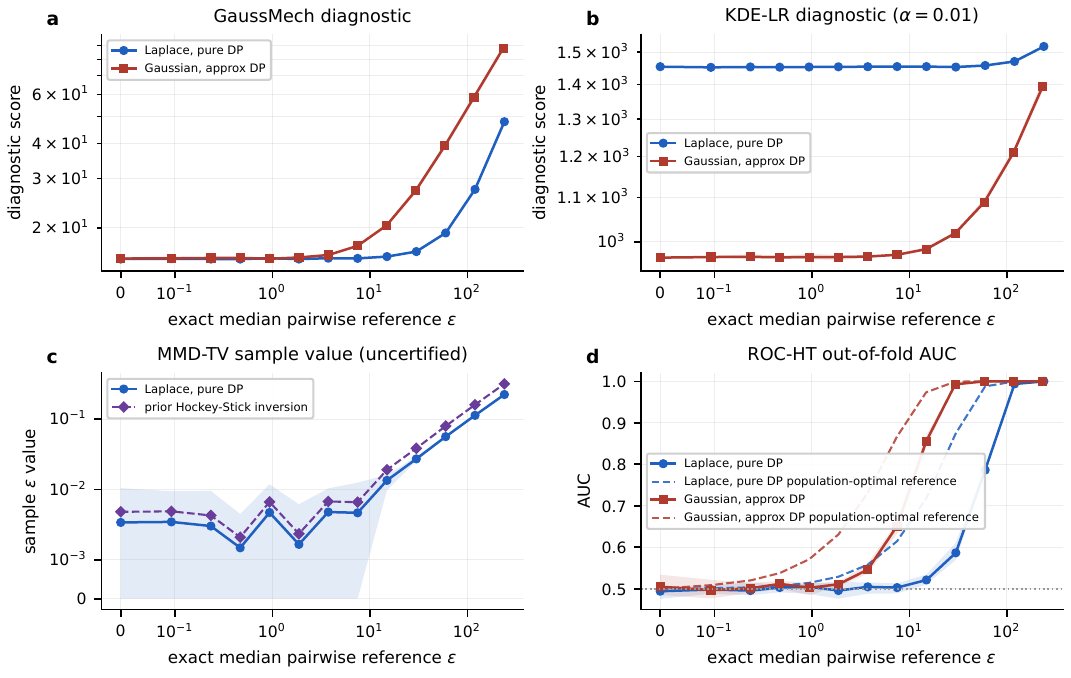}
\caption{Method outputs versus the known median pairwise $\varepsilon$ at
$n=100$. Lines show means; bands include both prototype sets and all
repetitions. The two diagnostic scores remain large when the true value is
zero. MMD-TV and the tighter earlier Hockey-Stick formula are shown only for
the pure-DP mechanism; both use sample estimates rather than confidence bounds.
ROC-HT is compared with the best possible Gaussian AUC and a Monte Carlo
approximation based on 20,000 draws per class for product Laplace noise.}
\Description{Four line-chart panels, one per audit method: GaussMech
diagnostic, KDE-LR diagnostic, MMD-TV sample value, and ROC-HT out-of-fold AUC.
Each plots its output against exact median pairwise epsilon at sample size 100,
with the Gaussian and Laplace mechanisms drawn as separate colored lines within
each panel; the ROC panel also shows population-optimal references.}
\label{fig:regime-response}
\end{figure*}

Figure~\ref{fig:regime-response} shows how each output changes as the known
privacy loss increases. The same pattern appears for both prototype sets:
\begin{itemize}
\item GaussMech is responsive once signal dominates its null floor, but its
value is not calibrated to the exact privacy reference.
\item KDE-LR changes with the noise level but remains dominated by its
high-dimensional finite-sample baseline and smoothing choice.
\item MMD-TV generally increases for the pure-DP mechanism, but at nonzero
settings its median sample value is only about $9.4\times10^{-4}$ of the exact
pairwise $\varepsilon$. Using the tighter earlier Hockey-Stick formula on the
same sample MMD raises the ratio only to about $1.3\times10^{-3}$. These ratios
include all sources of loss in the implemented calculation: the population
inequality, the kernel and bandwidth, and finite-sample estimation. They do not
measure the looseness of either population inequality alone. Among the 3,060
nonzero Laplace results with exact pairwise $\varepsilon\leq1$, the MMD-TV
sample value exceeds the truth 18 times, by as much as $1.95\times$. This shows
directly why the sample value is not a confidence bound.
\item Learned ROC AUC approaches the population-optimal reference for sufficiently
separated Gaussian mechanisms. For product Laplace noise, the gap can remain
large in 512 dimensions, showing that classifier family and sample complexity
matter in addition to $\varepsilon$.
\end{itemize}

\subsection{When Does the Confidence Bound Become Informative?}
\label{sec:known-cert}

For this more expensive validation, we use the ten target values
$\{0,1,2,4,8,16,32,64,128,256\}$ from the broader grid above.
For each predeclared pair, we split each identity's samples equally into
development and test sets. On the development data, we fit scikit-learn
logistic regression with $C=1$, the \texttt{lbfgs} solver, and
\texttt{max\_iter}$=1000$. We also select its orientation and threshold on
those data. We do not
tune this classifier on the benchmark, so the results describe this fixed
classifier rather than the best classifier that might be found. For comparison,
we also use the exact likelihood-ratio score, choosing only its threshold on
development data. This exact score is an oracle reference: it knows the
mechanism's likelihood ratio and is not a deployable auditor. Both fixed tests
are evaluated once on the unused data.

Let $U_{\mathrm{FP}}$ and $U_{\mathrm{FN}}$ be one-sided exact
Clopper--Pearson upper bounds. Within one combination of mechanism, sample
size, target $\varepsilon$, repetition, and prototype set, we test 40 pairs and
two error rates. We therefore use $\alpha=0.05/(40\cdot2)$ for each error rate.
The resulting 40 privacy bounds hold together with at least 95\% confidence,
separately for the learned classifier and exact likelihood-ratio score. We
report
\begin{equation}
\varepsilon_L=\max\!\left\{0,
\ln\frac{1-\delta-U_{\mathrm{FP}}}{U_{\mathrm{FN}}},
\ln\frac{1-\delta-U_{\mathrm{FN}}}{U_{\mathrm{FP}}}\right\},
\label{eq:heldout-lb}
\end{equation}
omitting terms with nonpositive numerators.

Across 14,400 pairwise evaluations, divided equally between the dispersed and
clustered prototype sets, no lower bound from either test exceeds the exact
pairwise value. A failure was numerically possible in 3{,}525 evaluations per
test: 1{,}605 for the dispersed set and 1{,}920 for the clustered set. Among
these cases, the observed failure rate is zero, with a 95\%
Clopper--Pearson upper bound of $8.5\times10^{-4}$. This checks the
implementation; the statistical guarantee comes from evaluating fixed tests
on unused data, not from observing zero failures.

The 95\% guarantee applies separately to each of the 360 experimental
settings, where one setting contains 40 pairs and two error rates. It does not
say that all 360 settings hold together with 95\% confidence. Even if every
setting has exact 95\% confidence, failures could occur somewhere in about 18
of the 360 settings. Section~\ref{sec:experiments} also uses
$\alpha=0.000625$, but for a different collection of tests.

\begin{figure*}[t]
\centering
\includegraphics[width=\textwidth]{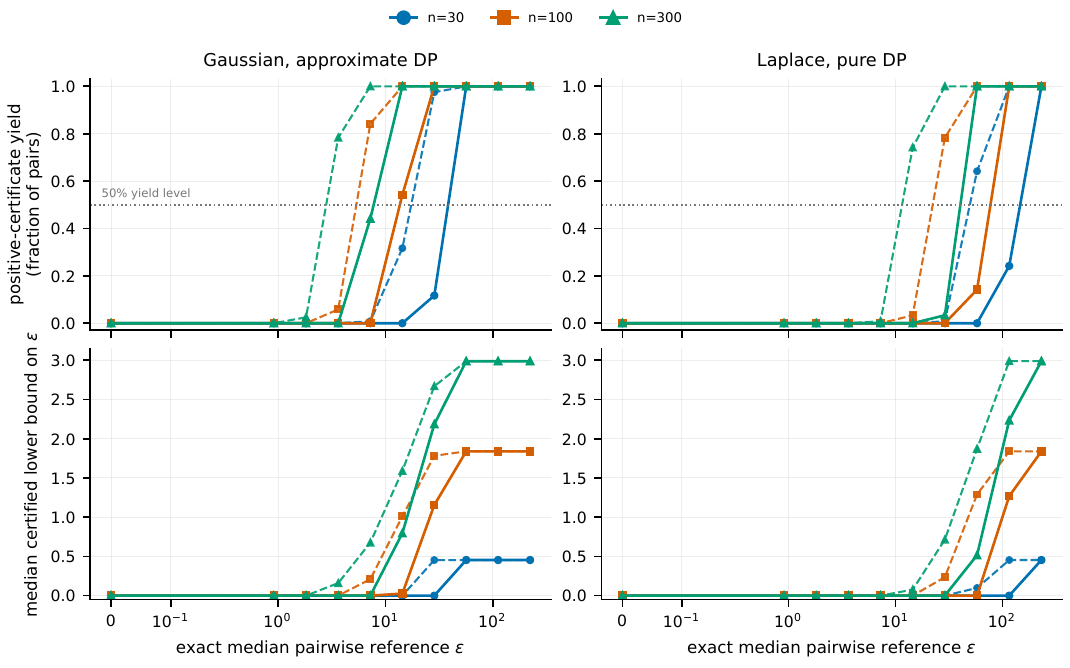}
\caption{How often the test gives a lower bound above zero for the dispersed
prototype set. Solid lines use logistic regression; dashed lines use the exact
likelihood-ratio score with a threshold selected on development data. More
samples increase both the fraction of positive bounds and their largest
possible value.}
\Description{A two-by-two grid of line charts. Columns separate the
Gaussian and Laplace mechanisms. The top row plots the fraction of pairs with
a positive lower bound, and the bottom row plots the median lower bound on
epsilon, both
against the exact median pairwise reference epsilon on a logarithmic axis. The
three sample sizes are distinguished by line color rather than by panel;
learned logistic tests are solid and exact likelihood-ratio tests are dashed.}
\label{fig:certificate-power}
\end{figure*}

Figure~\ref{fig:certificate-power} shows that a valid method need not produce
an informative result. Across $n=30,100,300$, the first tested global
$\varepsilon$ for which the learned classifier gives a positive bound on at
least half the pairs, averaged over three repetitions, is
$\{64,16,16\}$ (Gaussian) and $\{256,128,64\}$ (Laplace) under both
prototype sets. The exact likelihood-ratio score reaches that point at smaller
values. These thresholds apply only to the tested grid and identity pairs.

Because these thresholds average only three repetitions, we perform a larger
power experiment on five pairs from the dispersed prototype set. The pairs are
chosen without using the outcomes: two attain the sensitivity maxima, and
three lie at the 25th, 50th, and 75th distance percentiles. We repeat the full
procedure 200 times per pair, giving 24{,}000 replications; the artifact
includes an exact binomial confidence interval for each setting. The
probability that the learned classifier gives a positive bound rises
$0\to0.629\to1.000$ (Gaussian target $\varepsilon_g=16$) and
$0\to0.213\to1.000$ (Laplace target $64$) as $n$ goes $30\to300$. No
replication's lower bound exceeded its exact reference, but at these large
targets every attainable zero-error bound lies below its reference, so this
check was structurally incapable of failing; it is reported as an
implementation check only, not as evidence for the confidence guarantee.

\section{Discussion}
\label{sec:discussion}

The four columns in Table~\ref{tab:comparative} are not four estimates of the
same $\varepsilon$. \textsc{GaussMech} and \textsc{KDE-LR} depend on modeling
assumptions that do not hold for the full embedding distribution.
\textsc{MMD-TV} begins with a valid population inequality, but the reported
sample value has no one-sided confidence correction. The out-of-fold
\textsc{ROC-HT} value is also descriptive and often becomes infinite at a
finite-sample endpoint. Only the separate classifier evaluated on unused data
produces a lower confidence bound.

This distinction explains why the face-data values cannot be ranked by size.
On the same images, \textsc{KDE-LR} reports values in the thousands,
\textsc{GaussMech} reports values in the tens, \textsc{MMD-TV} remains below
one, and the direct ROC conversion is infinite. The disagreement reflects how
the methods are defined, not four different measurements on a common scale.
The known-$\varepsilon$ experiments provide the meaningful comparison: they
show whether each method is zero at the null, responds to increasing privacy
loss, and supports the claimed statistical interpretation.

The practical choice follows from the question being asked. \textsc{GaussMech}
is inexpensive. It may be useful when additive isotropic Gaussian noise is a
credible model. \textsc{KDE-LR} reveals coordinate-level differences but says
nothing rigorous about the full vector. \textsc{MMD-TV} avoids a classifier and
has a valid population interpretation, but it needs a finite-sample lower bound
before it can support a confidence claim. \textsc{ROC-HT} can use a strong
classifier; when that classifier and its threshold are fixed before test data
are examined, the resulting error bounds provide a valid finite-sample test.
The cost is low power when few test samples are available.

\section{Related Work}
\label{sec:related}

\noindent\textbf{DP auditing.} StatDP searches for
counterexamples~\cite{ding2018statdp}; DP-Sniper and Tight Auditing give
learned, confidence-treated attacks~\cite{bichsel2021dpsniper,nasr2023tight};
DP-Auditorium supplies a neighbor/tester abstraction with finite-sample MMD
machinery~\cite{kong2024dpauditorium}; sequential MMD is
anytime-valid~\cite{gonzalez2025sequential}; Bayesian estimation tightens
attack-based bounds~\cite{zanella2023bayesian}. Further work studies DP-SGD,
group violations, and audit information cost~\cite{jagielski2020auditing,
nasr2021adversary,steinke2023privacy,lokna2023groupattack,xiang2025bits}. We do
not introduce a general auditor, privacy mechanism, or confidence theorem. Our
contribution is the controlled comparison against exact privacy values.
DP-Auditorium~\cite{kong2024dpauditorium} and sequential MMD
~\cite{gonzalez2025sequential} provide valid tests in settings where the true
$\varepsilon$ is unknown. We instead ask whether several reported quantities
behave as expected when the pairwise value is known, including at an exact
null and at three sample sizes. Xiang et al.~\cite{xiang2025bits} bound the
information available to any audit; we measure how far these particular
methods remain from known targets.

\noindent\textbf{Generative privacy.} Extraction and membership inference test
diffusion-model memorization~\cite{carlini2023diffusion,hupang2023diffusion};
DP-SGD protects training data~\cite{abadi2016dpsgd}; and local or feature-space
mechanisms perturb one image or representation~\cite{xue2023dpimage,
wen2021identitydp,shibata2024dpddpm}. Our inference-time channel instead
specializes distribution privacy to identity-conditioned populations
~\cite{chen2023distributionprivacy}; retrospective \textsc{GaussMech} or
\textsc{KDE-LR} inherits no guarantee from these works.

\noindent\textbf{MMD and DP testing.} MMD/IPM--DP relations are established
~\cite{gretton2012kernel,sriperumbudur2010hilbert,kairouz2015composition,
ghazi2024totalvariation,dong2022gaussian}, including finite-sample kernel
auditors~\cite{domingo2022auditing,kong2024dpauditorium,gonzalez2025sequential};
our experiments show why inserting a sample MMD directly into a population
formula does not produce a finite-sample confidence bound.

\section{Limitations and Conclusion}
\label{sec:limitations}

The controlled experiments cover finite prototype sets and four mechanism
families (Appendices~\ref{app:strengthening} and~\ref{app:anisotropic}). Their
curves need not carry over to normalized embeddings or arbitrary generators.
The original face experiments also depend on an iid row assumption that their
saved generation records cannot verify. They do not study a fixed source image
or repeated releases. Appendix~\ref{app:random-cascade} separately studies
changing donors across five FaceFusion passes, but it does not remove the other
limitations listed in Appendix~\ref{app:limitations}.

Within these limits, the known-$\varepsilon$ experiments reveal errors that the
face experiments alone would hide. Two diagnostic scores are large even when
the true value is zero. The direct empirical ROC conversion often becomes
infinite at the same null. The MMD population formula is valid, but its sample
substitution is not a confidence bound. Evaluating a fixed classifier on unused
data avoids these problems and gives a lower confidence bound. Only this
reserved-test route yields a finite-sample certificate, but in most settings
its lower bound is zero and therefore uninformative. Every audit number should
retain its formal-status label and be reported with the assumptions and uncertainty
that give it meaning. Differential-privacy claims should be reserved for
certified results.

\clearpage

\begin{ethics}
This work measures identity leakage and therefore involves sensitive biometric
derivatives and dual-use risk. It recruited no participants, collected no new
face data, and conducted no user study. For the original FaceFusion and
InstantID experiments, the authors generated the saved embedding tables from
CelebA and original VGGFace2 images. The VGGFace2 inputs were used without face
restoration or enhancement, and the generated outputs were encoded with
ArcFace and FaceNet. We do not redistribute those tables. The manifested
fixed-donor control and random-donor cascade were likewise generated locally
from original VGGFace2, with their filtering and generation settings recorded.
The research team determined that ethics review was not required for this
work. Under 45 CFR
46.102(e) it is not human subjects research: we neither interact nor intervene
with any individual and obtain no identifiable private information, as the
corpus ships images under opaque numeric identifiers that we never attempt to
link to names. We obtained VGGFace2 from a public third-party dataset
repository that distributes it under a stated CC BY-NC 4.0 license. That is a
redistribution: the original maintainers have withdrawn their own distribution,
and prior public availability does not remove consent, licensing,
demographic-bias, or misuse concerns. We cannot establish
consent and licensing for the FaceFusion donor used in the original
experiments, so we make no such claim and do not redistribute that asset. The artifact excludes
source and generated faces, identity-level embedding rows, and face-derived
prototypes; it contains synthetic benchmarks, code, pseudonymous pair
identifiers, and aggregates. The embedding tables carry no
demographic labels, and joining them to source-dataset attributes (CelebA's
40 binary attributes; VGGFace2's gender field) would require image-level
provenance the repository does not preserve, so we could not perform the
subgroup analysis we would want. Relatedly, the clustered-prototype result
(Appendix~\ref{app:strengthening}) shows one global guarantee can induce
sharply different pairwise audit difficulty; combined with documented demographic
error differentials in face recognition, audit assurance itself may be
distributed unevenly across groups, and quantifying that disparity is future
work. The intended use is defensive evaluation of privacy claims.
\end{ethics}

\begin{openscience}
Upon acceptance, the authors plan to release an artifact containing the current audit and benchmark source
code, environment specification, unit tests, completed run manifests, shared
pair lists, aggregate result tables, and scripts that regenerate every
benchmark figure and table. The manifests record the source-file and generating
script hashes present at run time. Later integrity-only revisions were checked
against deterministic subsets of the completed runs; the artifact will report
that compatibility check rather than claim that an unarchived earlier script
snapshot is present. Because the artifact does not redistribute biometric
inputs or identity-level embeddings, it will not claim self-contained
end-to-end reproduction of the original FaceFusion or InstantID generation.
It will include a synthetic smoke test, the complete known-$\varepsilon$
outputs, and the fixed-pair power and null-floor sweeps reported here.
\end{openscience}

\begin{ai}
The authors used AI-based tools to assist with editing the manuscript text and
implementing portions of the benchmark and analysis code. AI-based tools were
also used to render the non-data-bearing schematic in
Figure~\ref{fig:pipeline}. The authors reviewed and verified all AI-assisted
content and take full responsibility for the accuracy, originality, and
integrity of this paper and its artifacts.
\end{ai}

\bibliographystyle{ACM-Reference-Format}
\bibliography{references}

\appendix

\section{Detailed Limitations}
\label{app:limitations}

\noindent\textbf{Saved data.}
The case study uses saved FaceFusion and InstantID outputs from selected
VGGFace2 and CelebA identities. The repository lacks exact generator revisions,
prompts, seeds, checkpoint hashes, and the FaceFusion donor identifier, so it
cannot reproduce image generation end to end. FaceFusion fixes one donor and
varies target images; its labels therefore denote target identities. Filtering
removes missing and zero-norm vectors, after which two CelebA FaceFusion
conditions each retain two duplicated generated embeddings: one repeated within
identity 1175, and one identical vector appearing under two different identity
labels (1489 and 1490). The cross-identity duplicate bears directly on the iid
row model and on identity separability.

\noindent\textbf{A later fixed-donor experiment.} The original eight
conditions do not preserve complete generation settings. We therefore applied
the same analysis to a new fixed-donor experiment on original VGGFace2 images,
without face restoration, for which those settings were recorded. It uses its
own identity-pair list and appears as the final row of
Table~\ref{tab:comparative}. Its diagnostic results fall within the ranges of
the original experiments. Bondalakunta et al.\ study leakage and repeated face
swapping across several tools~\cite{faceswapleakage2027}.

\noindent\textbf{Representation.}
The encoders are proxies for identity, not ground truth, and exact checkpoint
identifiers are unavailable. Because InstantID conditions on an ArcFace-family
representation, FaceNet provides the cleaner check against representation-family
overlap. Agreement across two encoders does not show
that they capture every identity cue in the pixels.

\noindent\textbf{Threat model.}
Definition~\ref{def:id-dp} treats identity-conditioned image sampling as part
of a single-query mechanism. It does not cover a fixed user-selected
photograph, repeated releases, a changing FaceFusion donor, or per-image,
prompt, and dataset-level adjacency relations. Those settings require new
mechanism definitions and composition analyses.

\noindent\textbf{Statistical scope.}
The calculations on unused face data assume that the saved rows are iid
samples from each $P_{\mathrm{id}}$. The incomplete generation records cannot
verify this assumption. The reported numbers are therefore binomial confidence
bounds for selected difficult pairs under that model. They are not
unconditional guarantees for the generators, a representative sample of all
pairs, or an estimate of how common leakage is. The smaller CelebA test sets
give no positive bounds after correcting for all 40 pairs. MMD-TV would also
need a one-sided finite-sample or sequential treatment before it could provide
a confidence bound. Finally, the synthetic experiments cover finite prototype
sets and four mechanisms: Laplace, isotropic Gaussian, anisotropic Gaussian,
and $K$-ary randomized response. They do not cover arbitrary generative
distributions.

\section{Additional Known-\texorpdfstring{$\varepsilon$}{Epsilon} Experiments}
\label{app:strengthening}

This appendix adds two privacy mechanisms to the benchmark: randomized
response and an anisotropic Gaussian mechanism. It also compares our test with
DP-Sniper. The completed randomized-response
experiment contains 4{,}320 rows. The DP-Sniper experiment contains 324 rows
and exactly 526{,}569{,}660 mechanism calls. The files were fixed before the
analysis and are checked by their hashes. The 95\% confidence statement for our
test applies separately to each combination of prototype set, requested
$\varepsilon$, sample size, and repetition. Each DP-Sniper bound applies to one
fixed identity pair and one direction.

\noindent\textbf{Randomized response.}
The first additional mechanism does not add continuous noise. It has $K=400$
possible identity labels. Given identity $i$, it reports $i$ with ideal
probability $e^{\varepsilon}/(e^{\varepsilon}+K-1)$; otherwise it chooses one
of the other labels uniformly and releases that label's prototype
~\cite{warner1965randomized,kairouz2016extremal}. The implementation represents
this probability by an integer threshold on a $2^{53}$ grid
~\cite{mironov2012significance,canonne2020discrete}. We therefore compare the
audits with the pure-DP $\varepsilon$ produced by that rounded threshold, not
only with the requested value.

Across 4{,}320 results, neither the learned classifier nor the exact test gives
a lower bound above the known pairwise value. A violation was numerically
possible in 1{,}440 rows for each test; with none observed, the 95\%
Clopper--Pearson upper bound on that rate is $2.1\times10^{-3}$. As before,
this checks the implementation rather than proving the confidence result.

The exact test gives no positive bounds at the null or at requested
$\varepsilon\in\{0.5,2\}$. The fraction becomes 0.75 at
$\varepsilon=8,n=30$, and 1.00 at $\varepsilon=8,n\geq100$ and at
$\varepsilon\in\{16,32\}$. The learned classifier differs by at most 0.005.
The zero result at $\varepsilon=2$ was expected. The total-variation distance
is $(e^2-1)/(e^2+399)\approx0.0157$, so a balanced test needs roughly
$1/\mathrm{TV}^2\approx4\times10^3$ samples per class. Our largest sample size
is 300.

The same finite-sample problems seen earlier also appear here. At every
nonzero setting, the \textsc{MMD-TV} median is far below the realized
$\varepsilon$. At the exact null, 318 of 720 MMD sample values are positive and
60\% of the direct ROC conversions are infinite because of zero denominators.
These values follow their formulas, but they are not evidence of positive or
infinite population privacy loss.

\noindent\textbf{Comparison with DP-Sniper.}
We also adapt DP-Sniper~\cite{bichsel2021dpsniper,zanella2023bayesian} to three
identity pairs selected before the experiment. With $6.5\times10^6$ mechanism
calls per result, it gives valid but conservative lower bounds
(Table~\ref{tab:dpsniper}) and gives no positive result at the exact null. When
restricted to the same $2n$ mechanism calls used by our classifier test, its
bound is zero or negative in most settings. These are per-pair comparisons
under different procedures and computational budgets, not an overall ranking
of the two methods.

\noindent\textbf{Unequal noise across coordinates.} A fourth mechanism changes
the geometry while keeping the noise Gaussian. It is
$c_i+s\Lambda^{1/2}z$ with predeclared diagonal $\Lambda$ (condition number
$100$), $s$ calibrated once at the global whitened sensitivity
$D_\Lambda$. Each pair's exact value follows from its distance after
whitening. Across 4{,}320 results, neither test gives a lower bound above the
known value; a violation was numerically possible in 1{,}944 rows per test.
The exact score, which knows the whitening direction
$\Lambda^{-1}(c_a-c_b)$, is positive for 0.64 of the pairs at
$\varepsilon_g=8,n=100$. The learned classifier is positive in only three of
the eighteen $(\varepsilon_g,n)$ settings: 0.21 at $(16,300)$, 0.02 at
$(32,100)$, and 0.88 at $(32,300)$. Here the main difficulty is learning the
right classification direction, not the confidence correction
(Table~\ref{tab:anisotropic}).

\noindent\textbf{Nearby identities are harder to audit.} In the main
known-$\varepsilon$ experiment, the clustered prototype set contains 10 pairs
within a cluster and 30 pairs from different clusters. They have
median exact references $2.021$ and $15.397$ at $n=100,\varepsilon_g=16$
(Gaussian). Their median learned/exact-test lower bounds are $0/0$ within a
cluster and $0.156/1.082$ between clusters. Thus one global privacy guarantee
can produce identity pairs with very different audit difficulty.

For the exact test, the fraction of positive bounds is identical for the
dispersed and clustered prototype sets at every grid point except requested
$\varepsilon=8,n=30$: 0.758 for the dispersed set and 0.733 for the clustered
set. The geometry does not change the exact privacy value, but it has a small
effect on detection with few samples.

For DP-Sniper, the medians under the 6.5-million-call setting differ only in
the second decimal among the minimum, upper-median, and maximum $L_2$ pairs
($4.483/4.464/4.471$ for randomized response at requested $\varepsilon=8$),
consistent with the family's pair-uniform realized epsilon.

At the same total sample budget as our classifier test, $2n$ mechanism calls
split across DP-Sniper's four sample pools, the
median DP-Sniper lower bound is $-\infty$ or negative in most settings; for
randomized response at requested $\varepsilon=32$ it reaches $0.54/1.96/3.12$
at $n=30/100/300$, still far below that family's realized reference of
$32.0$.

\begin{table}[t]
\centering
\small
\caption{DP-Sniper results using $6.5\times10^6$ mechanism calls per row.
The full experiment has 81 rows at this budget and 243 rows at the matched
$2n$ budget, for 526{,}569{,}660 calls. Entries are median 95\% lower bounds
over three fixed pairs and three repetitions. Ratios compare each bound with
the exact pairwise value, using unrounded numbers. Each result applies to one
fixed pair and direction and to pure DP. None of the 324 bounds exceeds its
exact reference.}
\label{tab:dpsniper}
\begin{tabular}{@{}lcccc@{}}
\toprule
requested $\varepsilon$ & 0.5 & 2 & 8 & 32\\
\midrule
product-Laplace bound & $-0.03$ & $0.12$ & $0.75$ & $2.63$\\
\quad ratio to exact reference & $-0.05$ & $0.06$ & $0.09$ & $0.08$\\
$K$-ary RR bound & $0.12$ & $0.93$ & $4.47$ & $4.60$\\
\quad ratio to realized reference & $0.23$ & $0.46$ & $0.56$ & $0.14$\\
\bottomrule
\end{tabular}
\end{table}

\noindent\textbf{DP-Sniper settings.} We use the published
DD-Search/DP-Sniper implementation with its PyTorch logistic-regression attack.
The optimizer is stochastic gradient descent (SGD), with learning rate $0.3$,
momentum $0.3$, and a \texttt{StepLR} update at step $500$. We use attack-region quantile
$c=0.01$, prediction/training batch sizes $10^{4}/10^{5}$, a single worker
process, input arity one (the fixed pair mapped to inputs $\{0,1\}$), final
evaluation on independent samples with no direction swap, and distinct
per-stage seeds for attack training, orientation, and final evaluation. We do
not use the implementation's input-search component because the identity pair
is fixed before each run. The comparison is therefore limited to certification
for a fixed pair.
The $c=0.01$ attack-region quantile together with the finite final pool caps
the attainable bound well below $32$ regardless of the mechanism, so the
flat randomized-response medians at requested $\varepsilon=32$ reflect the
tool's configured ceiling, not a mechanism property.

\section{Anisotropic-Gaussian Details}
\label{app:anisotropic}

This experiment has 4{,}320 results: two prototype sets, six requested global
$\varepsilon$ values, three sample sizes, three repetitions, and 40 shared
pairs. The completed run is identified by a file hash that the analysis checks
before recomputing the results.

\noindent\textbf{Noise shape.} The diagonal matrix $\Lambda$ determines how
much noise is added to each coordinate. Its entries are
$\Lambda_{kk}\propto C^{1/2-k/(d-1)}$, in descending order, with condition
number $C=100$ and geometric mean one. We choose this deterministic,
data-free matrix before running the experiment and record its hash. After the
coordinates are whitened by $\Lambda^{-1/2}$, the mechanism becomes the usual
isotropic analytic Gaussian mechanism. We therefore calibrate one scale $s$
using the largest whitened distance
$D_\Lambda=\max_{i\neq j}\lVert c_i-c_j\rVert_{\Lambda^{-1}}$
and compute every pair's exact reference from its own whitened distance at
$\delta=10^{-5}$. These distances range from $1.745$ to $2.311$ in the
dispersed prototype set and from $0.348$ to $2.108$ in the clustered set.

No learned or exact-test lower bound exceeds its known pairwise value. Such a
failure was numerically possible for 1{,}944 rows per test. With no failures
observed, the 95\% Clopper--Pearson upper bound on that rate is
$1.5\times10^{-3}$.

At the exact null, neither classifier test gives a positive bound in any of
the 720 rows. In contrast, 340 MMD sample substitutions are positive and 546
direct out-of-fold ROC conversions are infinite because of observed zero
denominators. Thus unequal coordinate variances do not remove the finite-sample
problems found in Section~\ref{sec:known-eps}.

The noise geometry makes the correct classification direction harder to learn.
At $\varepsilon_g=8,n=300$, the exact test gives a positive bound for every
pair in the dispersed set and 0.750 of the pairs in the clustered set, while
the learned classifier gives none. At $\varepsilon_g=32,n=300$, the learned
fractions rise to 1.000 and 0.750. Even when a bound is positive, its median is
at most 0.11 of the exact pairwise $\varepsilon$ for the exact test and 0.022
for the learned classifier.

\begin{table}[t]
\centering
\small
\caption{Fraction of the 40 anisotropic-Gaussian pairs with a lower bound above
zero, averaged over both prototype sets and three repetitions. The exact test
uses the mechanism's whitening direction and serves as an upper reference; the
learned test uses logistic regression fitted on development data.}
\label{tab:anisotropic}
\begin{tabular}{@{}lcccccc@{}}
\toprule
& \multicolumn{6}{c}{requested global $\varepsilon$}\\
\cmidrule(l){2-7}
test, $n$ & 0 & 0.5 & 2 & 8 & 16 & 32\\
\midrule
exact, 30 & 0.00 & 0.00 & 0.00 & 0.01 & 0.29 & 0.85\\
exact, 100 & 0.00 & 0.00 & 0.00 & 0.64 & 0.88 & 0.88\\
exact, 300 & 0.00 & 0.00 & 0.04 & 0.88 & 0.90 & 0.98\\
\midrule
learned, 30 & 0.00 & 0.00 & 0.00 & 0.00 & 0.00 & 0.00\\
learned, 100 & 0.00 & 0.00 & 0.00 & 0.00 & 0.00 & 0.02\\
learned, 300 & 0.00 & 0.00 & 0.00 & 0.00 & 0.21 & 0.88\\
\bottomrule
\end{tabular}
\end{table}

\noindent\textbf{Records needed for a complete face-generator audit.} For each
released image, the experiment should save the generator name and revision,
model checkpoint hashes, prompts and configuration files, sampling seeds, the
donor identifier for face swaps, the encoder version and checkpoint, and a
link from each embedding to the image-generation event. It should also state
whether each generated image uses an independent draw or reuses an earlier
input. The original eight face conditions preserve only the encoder family and
dataset labels from this list.

\section{Method Settings}
\label{app:route-config}

\noindent\textbf{\textsc{GaussMech} configuration.} We set $\delta=10^{-5}$, compute the distance between the two sample mean embeddings, and divide by the larger of their two root-mean-square residual scales as specified in Equation~\ref{eq:gaussmech}.

\noindent\textbf{\textsc{KDE-LR} configuration.} Identity-specific Silverman bandwidths are used to evaluate Gaussian KDEs on 201-point grids. Each grid spans the joint sample range plus three times the larger bandwidth; densities are floored at $10^{-12}$. We report $\alpha\in\{0,10^{-2}\}$.

\noindent\textbf{\textsc{MMD-TV} configuration.} We use an RBF kernel and a
per-pair median-heuristic bandwidth. Equation~\ref{eq:mmd-unbiased} is evaluated
at that data-selected bandwidth; negative values are clipped before taking the
square root. Permutation tests use $10^4$ permutations on the first 100 pairs
per condition with the Phipson--Smyth correction~\cite{phipson2010permutation}.

\noindent\textbf{\textsc{ROC-HT} configuration.} Per identity pair, we train scikit-learn's \texttt{LogisticRegression}~\cite{pedregosa2011scikit} with $C = 1.0$, the \texttt{lbfgs} solver, and \texttt{max\_iter} $= 1000$ on the L2-normalized 512-dimensional embeddings, under $5$-fold \texttt{StratifiedKFold} cross-validation with a fixed random seed. ROC curves are built from the concatenation of out-of-fold predicted probabilities across all $5$ folds.

\section{Repeated Random-Donor Cascade}
\label{app:random-cascade}

Motivated by the five-pass identity-dilution experiment of Bondalakunta et
al.~\cite{faceswapleakage2027}, we study what happens when FaceFusion is
applied five times to the same output. Their data contain one such sequence per
target identity. A statistical audit needs many outputs for each identity, so
we generate 200 separate five-pass sequences for every target. Each pass uses
a newly sampled donor whose distribution is independent of the target. We use
the same development/reserved-test construction as
Section~\ref{sec:known-cert}. This is an additional check on a different
face-generation process, not an independent validation of our method.

\noindent\textbf{How the images are generated.} We use 100 target identities.
For each one, we generate 200 five-pass sequences and divide them before
analysis into 100 development sequences and 100 test sequences. A sequence
starts from an image sampled with replacement from the target's fixed pool of
100 images. It then samples five different donor identities from a fixed pool
of 600 and uses one donor at each pass. The donor-selection rule is the same
for every target identity.

At each pass, the face encoder either returns a normalized 512-dimensional
embedding $\widehat z$ or finds no face. We store either outcome as a
513-dimensional vector: the first as $[\widehat z,0]$ and the second as
$[\mathbf{0}_{512},1]$. Thus the last entry
states whether detection failed; it does not pretend that a failed image has a
face embedding. Image generation completed in every case, but face detection
failed for 19 development outputs and 22 test outputs. Dropping these images
would analyze only successful detections and could hide an identity-dependent
failure pattern. The four descriptive methods use the development data. Pair
selection, classifier fitting, and threshold selection also use only the
development data; each fixed test is then evaluated once on the unused data.

\noindent\textbf{Why the table contains 50 pairs.} The first analysis included
five identity pairs. After inspecting those results, we used development data
from 90 identities whose test samples had not been examined to select 45 more
non-overlapping pairs. We then evaluated their test samples once. Finally, we
recomputed all original and added error bounds together using one Bonferroni
correction over 500 error rates. The individual confidence bounds remain valid,
but the decision to enlarge the analysis was not made in advance. Because the
additional pairs were deliberately chosen using the development data, this
50-pair panel is a stress test, not a random sample from which to estimate how
common the leakage is.

\begin{table}[t]
\centering
\small
\caption{Lower confidence bounds for 50 difficult, non-overlapping identity
pairs selected on development data. They are not an estimate of how common
leakage is. At $\delta=10^{-5}$, the 250 pair-pass bounds use 500 one-sided
error bounds with $\alpha=0.05/500=10^{-4}$. Under the stated iid model, all
results hold together with at least 95\% confidence. ``Censored'' counts pairs
with no errors in either direction. For those pairs, the finite sample size
limits the reported bound to 2.338. Because these are lower bounds and some hit
that limit, they cannot measure the true change in $\varepsilon$ across passes.}
\label{tab:random-cascade-cert}
\begin{tabular}{@{}rccccc@{}}
\toprule
Pass & Median $\varepsilon_{\mathrm{LB}}$ & Range & Positive & $>1$ & Censored\\
\midrule
1 & $2.290$ & $1.683$--$2.338$ & $50/50$ & $50/50$ & $9/50$\\
2 & $2.076$ & $1.367$--$2.338$ & $50/50$ & $50/50$ & $4/50$\\
3 & $1.890$ & $1.085$--$2.338$ & $50/50$ & $50/50$ & $2/50$\\
4 & $1.784$ & $1.025$--$2.338$ & $50/50$ & $50/50$ & $2/50$\\
5 & $1.701$ & $0.671$--$2.312$ & $50/50$ & $46/50$ & $0/50$\\
\bottomrule
\end{tabular}
\end{table}

\begin{table*}[t]
\centering
\small
\caption{Descriptive results on the development data after each FaceFusion
pass. The first five numeric columns are pairwise medians; the last three are
counts. GaussMech and MMD-TV use 500 pairs, KDE-LR uses 100, and the
out-of-fold ROC uses 200. These differ from the 50 pairs in
Table~\ref{tab:random-cascade-cert}. MMD permutation tests use the first 100
pairs at each pass. Every $p$-value is $1/10001$, which remains significant at
$0.05/500$. These descriptive calculations omit rows where face detection
failed. None is a finite-sample confidence bound on $\varepsilon$.}
\label{tab:random-cascade-diag}
\begin{tabular}{@{}rrrrrrrrr@{}}
\toprule
Pass & \textsc{GaussMech} & KDE $\alpha{=}0$ & KDE $\alpha{=}10^{-2}$ &
\textsc{MMD-TV} & OOF AUC & Exact corner & Raw $+\infty$ & MMD reject\\
\midrule
1 & $39.5$ & $5{,}408$ & $1{,}083$ & $0.181$ & $1.000$ & $79/200$ & $200/200$ & $100/100$\\
2 & $34.3$ & $5{,}211$ & $1{,}051$ & $0.153$ & $0.999$ & $43/200$ & $200/200$ & $100/100$\\
3 & $32.1$ & $5{,}166$ & $1{,}039$ & $0.141$ & $0.997$ & $25/200$ & $200/200$ & $100/100$\\
4 & $30.5$ & $5{,}080$ & $1{,}028$ & $0.131$ & $0.994$ & $6/200$ & $200/200$ & $100/100$\\
5 & $29.3$ & $5{,}034$ & $1{,}025$ & $0.124$ & $0.991$ & $4/200$ & $200/200$ & $100/100$\\
\bottomrule
\end{tabular}
\end{table*}

Observed distinguishability weakens with each pass. The number of ROC curves
that reach the exact zero-error corner falls from 79 of 200 after the first
pass to 4 of 200 after the fifth. The median reserved-test lower bound falls
descriptively from 2.290 to 1.701. Nevertheless, every one of the 50 selected
pairs still has a positive bound after the fifth pass. The direct ROC
conversion is less useful: it remains $+\infty$ for all 200 pairs at every
pass. These values are lower bounds, not estimates of the true
$\varepsilon$, and some early values are capped by the finite test size. The
experiment therefore does not establish whether the true $\varepsilon$
decreased, much less the size of any decrease.

\section{Additional Robustness Results}
\label{app:robustness}

\noindent\textbf{Gaussian residuals.}
We pool within-identity residuals across retained identities and inspect Q-Q
plots and Shapiro--Wilk tests~\cite{shapiro1965analysis}. The pooled rows are
not iid draws from one identity, so the tests are diagnostics rather than exact
model-selection tests. Rejection at $p<10^{-3}$ ranges from $296/512$
coordinates in the ArcFace, CelebA, InstantID condition to $512/512$ in the
ArcFace, VGGFace2, FaceFusion condition, and exceeds $300/512$ in seven of
eight conditions.

\begin{figure*}[t]
\centering
\includegraphics[width=\textwidth]{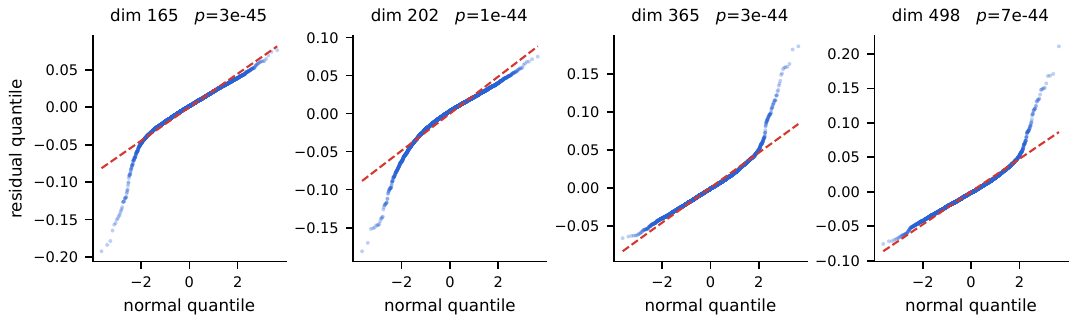}
\caption{Q-Q plots for the four least-Gaussian ArcFace dimensions under
VGGFace2+InstantID. Each panel uses the same 5,000-row subsample drawn from
approximately 40,000 pooled within-identity residuals.}
\Description{Four quantile-quantile plots compare pooled residual coordinates
with Gaussian reference lines and show systematic departures in the selected
least-Gaussian ArcFace dimensions.}
\label{fig:qqplots}
\end{figure*}

\noindent\textbf{KDE smoothing.}
Changing $\alpha$ from $0$ to $10^{-2}$ reduces the summed diagnostic by
$4.8$--$6.4\times$ across conditions. At $\alpha=0$, tail grid points with
near-zero fitted density dominate; the uniform mixture limits those ratios,
but neither setting has a guarantee for the joint release.

\noindent\textbf{MMD bandwidth and classifier comparison.}
For the ArcFace embeddings of VGGFace2 outputs, sweeping
$\sigma/\sigma^\star\in
\{0.5,\allowbreak 0.75,\allowbreak 1,\allowbreak 1.5,\allowbreak 2\}$ gives
maximum audited-pair MMD-TV estimates
$\{0.677,\allowbreak 0.724,\allowbreak 0.642,\allowbreak 0.480,\allowbreak
0.374\}$ for FaceFusion and
$\{0.413,\allowbreak 0.448,\allowbreak 0.406,\allowbreak 0.310,\allowbreak
0.244\}$ for InstantID. Across all conditions, the
mean ratio between the classifier-restricted
$\max_t(\mathrm{TPR}-\mathrm{FPR})$ and $\mathrm{MMD}/2$ is
$4.39$--$7.52$. Neither sample quantity equals population TV, so this does not
isolate the slack of Lemma~\ref{lem:mmdtv}.

\noindent\textbf{Classifier and null controls.}
A 100-tree random forest~\cite{breiman2001random} on identical pairs and folds changes mean AUC by at
most $0.014$ relative to logistic regression. Inter-identity MMD permutation
tests reject $799/800$ pairs in separate uncorrected tests at level $0.05$,
while same-identity split-half tests reject $16/400$ controls under the same
per-pair level; these counts are descriptive rather than family-wise
discoveries. Same-identity ROC controls have chance mean AUC and no exact
$(0,1)$ corner.

\begin{figure*}[t]
\centering
\includegraphics[width=0.92\textwidth]{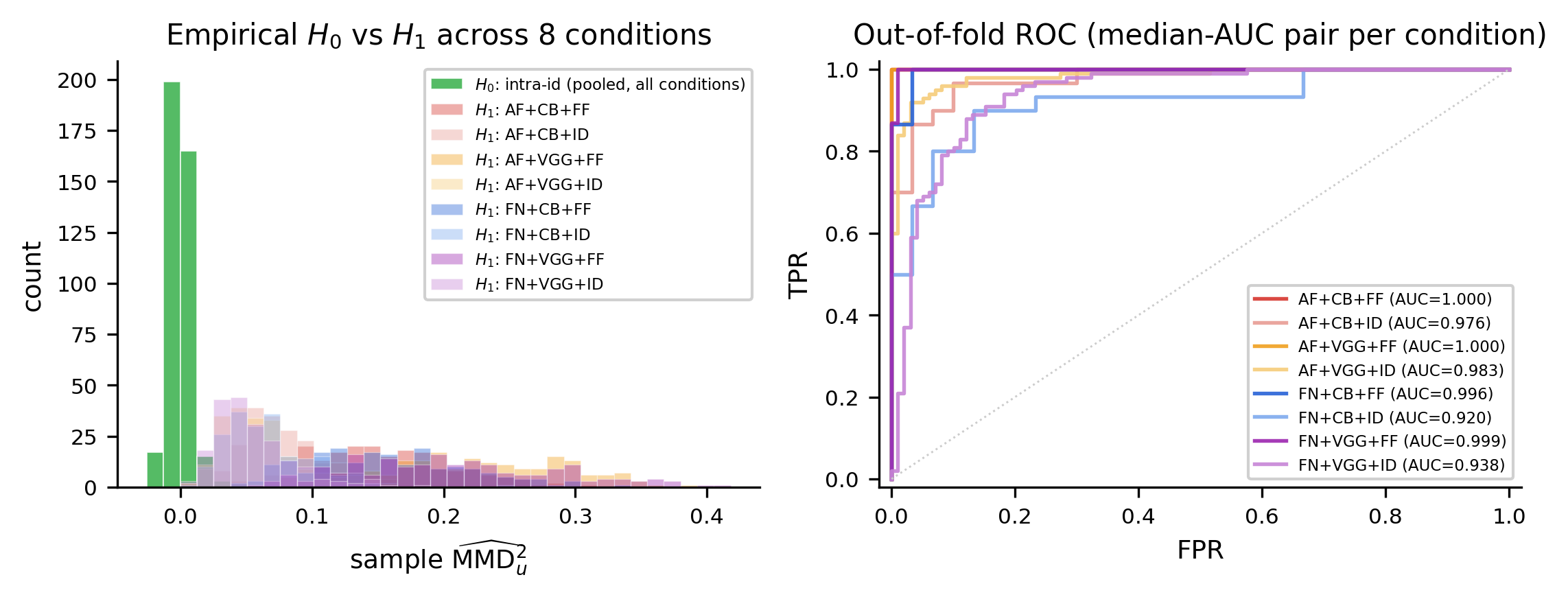}
\caption{\emph{Left}: real same-identity split-half and inter-identity
$\widehat{\mathrm{MMD}}^2_u$ values. \emph{Right}: one representative
median-AUC out-of-fold ROC per condition. Exact corner counts are reported in
the artifact.}
\Description{The left panel compares empirical same-identity and
inter-identity squared MMD distributions. The right panel overlays one
representative out-of-fold ROC curve for each encoder, dataset, and generator
condition.}
\label{fig:mmdroc}
\end{figure*}

\clearpage

\end{document}